\documentclass[11pt]{article}

\usepackage[margin=1in]{geometry}
\usepackage{setspace}
\usepackage{amsmath, amssymb, amsfonts, amsthm}
\usepackage{mathtools}

\usepackage{graphicx}
\graphicspath{{Figures/}}

\usepackage{booktabs}
\usepackage{caption}
\usepackage{subcaption}
\usepackage{tikz}
\usetikzlibrary{arrows.meta, positioning, shapes.geometric, fit}
\usepackage{xcolor}
\usepackage[percent]{overpic}
\usepackage[numbers]{natbib}
\usepackage[colorlinks=true, linkcolor=blue, citecolor=blue, urlcolor=blue]{hyperref}

\usepackage{algorithm}
\usepackage{algorithmic}

\newtheorem{theorem}{Theorem}
\newtheorem{proposition}{Proposition}
\newtheorem{lemma}{Lemma}
\newtheorem{corollary}{Corollary}
\newtheorem{assumption}{Assumption}

\newtheorem{remark}{Remark}
\usepackage{comment}

\newcommand{\E}{\mathbb{E}}
\newcommand{\Prob}{\mathbb{P}}

\title{Dynamic Persuasion: \\Managing Reputation under Censored Feedback}

\author{
Ahmet E. Arcaklioglu\thanks{Department of Industrial and Systems Engineering,
Texas A\&M University, College Station, TX, USA. Emails:
\texttt{arcaklia@tamu.edu}, \texttt{eksinc@tamu.edu}.}
\and
Ceyhun Eksin\footnotemark[1]
}

\date{}

\begin{document}
\maketitle
\begin{abstract}
We study a long-run sender who repeatedly communicates about a state to a
sequence of myopic receivers and whose credibility is a public score updated
from past interactions. Receivers face a lemon loss and an opportunity loss,
and their behavior depends on these only through a single index, the
tolerance. Because the score records only detected misrepresentations, it
rests above the sender's true honesty rate; we characterize this credibility
premium in closed form. In a fluid limit in which the score adjusts slowly,
no policy earns more than operating at a single credibility level, and the
sender's gain from managing credibility equals the value of the premium:
against receivers who correct for imperfect detection, that gain vanishes.
At a fixed adjustment speed, every policy is bounded by the fluid value plus
a term linear in the speed. We then study three policy classes, linear,
$(s,S)$, and threshold, obtain the optimal parameters of the linear and
threshold classes in closed form, show how receiver losses shape them, and
show that, under an approximation, the optimal $(s,S)$ band
collapses to a threshold. Simulations of the dynamics agree
with the closed forms and indicate where the single-level conclusion stops:
bands pay when switching modes is costly, and cycles pay when demand is
seasonal.
\end{abstract}


\section{Introduction}

Information design studies how a better-informed sender should disclose
information to shape the decisions of a less-informed receiver
\cite{kamenica2011bayesian, crawford1982strategic, akyol2016information}. Much
of this literature assumes that the sender can \emph{commit} to a signaling
rule, so that receivers act on messages as if their informational content were
guaranteed \cite{kamenica2011bayesian, little2023bayesian, sayin2021bayesian,
velicheti2023strategic, massicot2025almost}. In many ongoing relationships no
such commitment is available. The receiver cannot verify the sender's strategy
and instead judges each message through the lens of the sender's past behavior.
Credibility then becomes a state that the sender controls over time, and the
problem shifts from one-shot disclosure to the long-run management of trust
\cite{honryo2018dynamic, ball2023dynamic, ely2017beeps, best2024persuasion,
mathevet2018reputation}. This creates a dynamic trade-off: truthfulness
sacrifices immediate influence but builds credibility that raises the value of
future messages, while misrepresentation captures a short-run gain at the risk
of eroding the trust on which future persuasion depends \cite{ely2017beeps,
barisaac2003reputation, manapat2013information}.

Consider an online marketplace in which a long-run seller repeatedly sells
refurbished smartphones to a sequence of short-run buyers. The quality of each
phone varies: a higher-quality unit may have a healthy battery, original parts,
and no hidden water damage. The seller usually knows the condition before
listing it, but the buyer cannot verify it at the time of purchase. Instead, the
buyer observes the seller's claim, such as ``excellent condition,'' together
with the seller's reputation, and decides whether to buy. When the phone is low
quality, the seller can disclose the issue, which lowers the chance of an
immediate sale but protects its standing with future buyers, or misrepresent it
as high quality, which raises the chance of a sale but risks a loss of trust if
the buyer later discovers the mismatch. Some defects surface through returns,
negative reviews, third-party inspection, or warranty claims; others are never
noticed, or are noticed and not reported. Figure~\ref{fig:game_timeline_horizontal}
summarizes one period of this interaction.

\begin{figure}[h]
\centering
\resizebox{\linewidth}{!}{%
\begin{tikzpicture}[
    node distance=0.55cm and 0.65cm,
    >=Latex,
    every node/.style={font=\scriptsize},
    box/.style={
        rectangle,
        draw,
        rounded corners,
        align=center,
        minimum width=1.9cm,
        minimum height=0.75cm,
        thick
    },
    policybox/.style={box, fill=orange!20},
    processbox/.style={box, fill=blue!12},
    utilitybox/.style={box, fill=purple!12},
    smallbox/.style={
        rectangle,
        draw,
        rounded corners,
        align=center,
        minimum width=1.75cm,
        minimum height=0.65cm,
        thick
    },
    nobuybox/.style={smallbox, fill=red!12},
    buybox/.style={smallbox, fill=green!15},
    decision/.style={
        diamond,
        draw,
        aspect=1.7,
        align=center,
        thick,
        fill=yellow!20,
        inner sep=0.5pt
    },
    line/.style={->, thick}
]

\node[policybox] (policy) {Choose policy\\and parameter};

\node[processbox, right=of policy] (state) {\(s_t\) realized\\and observed};

\node[processbox, right=of state] (signal) {\(\theta_t\) sent};

\node[processbox, right=of signal] (receiver) {Receiver acts\\using prior, \(\theta_t\), \(r_t\)};

\node[decision, right=0.9cm of receiver] (buy) {Buy?};

\node[nobuybox, above right=0.45cm and 0.85cm of buy] (nobuy) {No purchase \\ Trust  updates};

\node[buybox, below right=0.45cm and 0.85cm of buy] (yesbuy) {Purchase\\trust updates};

\node[utilitybox, right=2.15cm of buy] (utility) {Utilities\\realized};

\node[processbox, right=of utility] (next) {Next period\\new receiver};

\draw[line] (policy) -- (state);
\draw[line] (state) -- (signal);
\draw[line] (signal) -- (receiver);
\draw[line] (receiver) -- (buy);

\draw[line] (buy) -- node[above, pos=0.55] {No} (nobuy.west);
\draw[line] (buy) -- node[below, pos=0.55] {Yes} (yesbuy.west);

\draw[line] (nobuy.east) -- ++(0.35,0) |- (utility.north);
\draw[line] (yesbuy.east) -- ++(0.35,0) |- (utility.south);

\draw[line] (utility) -- (next);

\draw[line] (next.south) .. controls +(0,-1.1) and +(0,-1.1) .. (state.south);

\node[
    draw,
    dashed,
    rounded corners,
    fit=(state)(signal)(receiver)(buy)(nobuy)(yesbuy)(utility)(next),
    inner sep=0.15cm,
    label=above:{\scriptsize Repeated interaction: period \(t\)}
] {};

\end{tikzpicture}%
}
\caption{Timeline of the repeated persuasion game. The sender chooses a policy class and parameter once and keeps it fixed. In each period, the sender observes the state $s_t$, sends a signal $\theta_t$, and the receiver decides whether to buy using the prior, the signal, and current trust $r_t$. Utilities are realized, trust is updated and the process repeats with a new receiver.}
\label{fig:game_timeline_horizontal}
\end{figure}
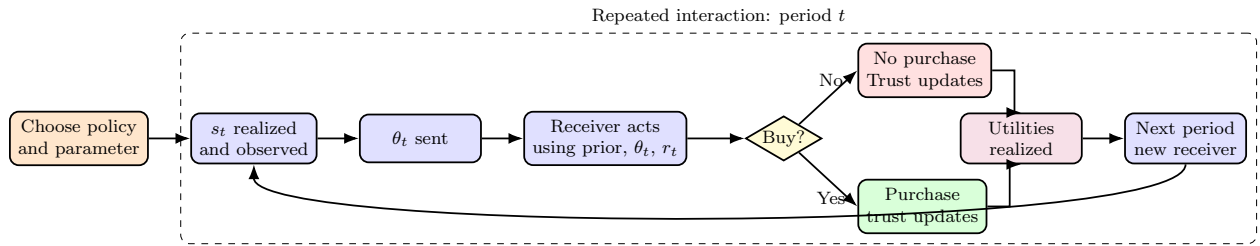

Truthful disclosure can also be implemented through the seller's pricing and
listing strategy. A phone with non-original replacement parts may be listed as a
lower-quality item at a lower price, and vice versa. Under this reading,
reporting low quality means forgoing the higher revenue that a high-quality
listing would have brought, so truthful disclosure may generate a lower profit,
or no sale, in the current period. This is an opportunity cost that the seller
accepts to preserve trust for future interactions. In the model below we take
the extreme form of this cost: a low signal yields no sale in the current
period, and the sender's per-period payoff is the sale indicator. A variant in
which a disclosed lemon still sells at a discount leaves the analysis unchanged
in structure, since the per-period reward would remain affine in the
misrepresentation rate.

In such marketplaces, buying a misrepresented lemon typically costs the buyer
far more than passing on a good unit. We capture this with a two-parameter
receiver loss structure: a \emph{lemon} loss $c_l$ from buying a low-quality
product and an \emph{opportunity} loss $c_o$ from declining a high-quality one.
This asymmetry changes the receiver's willingness to act on a favorable signal,
and therefore the credibility the sender must maintain to persuade.

Existing work on reputation for honesty studies fully rational receivers who
update beliefs about a sender's type from the complete, or from a limited,
record of play \cite{fudenberg2022honesty, jullien2014new, liu2014limited,
best2024persuasion, mathevet2018reputation}. Under imperfect monitoring such
reputations are impermanent: receivers eventually learn the sender's type and
the reputation disappears \cite{cripps2004imperfect}. On real platforms,
however, buyers typically act on an aggregate score rather than on a history,
and that score is built from a censored sample. A truthful disclosure is always
recorded, whereas a misrepresentation enters the record only if it is detected.
The censoring of platform feedback and the resulting inflation of scores are
documented empirically \cite{dellarocas2008silence, filippas2022}. This paper aims to characterize how a long-run seller should
manage such a score, how much the censoring is worth to the seller, and whether
it creates an incentive for build-and-exploit cycles.

We address these questions in a tractable model. The score follows an
exponentially weighted update of the kind used by online feedback systems, and a
heterogeneous population of myopic buyers responds to it at face value. The buyers read the posted score as the sender's honesty rate and best respond to it without conditioning on the fact that a misrepresentation is recorded only when it is detected. We do not model how buyers arrive at this rule, and we do not claim it is an equilibrium. Instead we make the assumption's cost explicit: we compare the sender's value against face-value buyers with its value against buyers who know the detection probability and correct for the censoring, and show that the entire gain from managing the score is the rent that censoring creates. Against buyers who correct for it, the sender's best policy is always honest or always deceptive, and credibility management is worthless (Sections~\ref{sec:premium} and~\ref{sec:fluid}).

Within this model we study three trust-dependent signaling policies: a
\emph{linear} policy, under which the probability of misrepresentation rises
smoothly with the score; a \emph{threshold} policy, which switches between
truth-telling and exploitation at a single cutoff; and an $(s,S)$ policy, which
builds credibility up to $S$ and spends it down to $s$. A natural conjecture is
that the $(s,S)$ band, which encodes an explicit build-and-exploit cycle,
outperforms a single cutoff. Under a regenerative-cycle approximation the optimal band collapses to a threshold, so the build-and-exploit cycle does not benefit the sender. This is not an artifact of the three classes. In the limit in which the score adjusts slowly, no policy, whether Markov, history-dependent or randomized, earns more than operating at a single credibility level, and the threshold rule attains that value. At a fixed adjustment speed every policy is
bounded by the fluid value plus a term linear in the speed. Simulations explore the extent to which a single credibility level remains optimal, showing that wide bands pay when switching modes is costly, and cycles pay when demand is seasonal.

\subsection{Related Work}

\paragraph{Persuasion and strategic communication.} A first strand studies
communication under commitment, from Bayesian persuasion
\cite{kamenica2011bayesian, little2023bayesian} to its control-theoretic
extensions with quadratic costs, multiple senders and learning
\cite{sayin2021bayesian, velicheti2023strategic, velicheti2024learning,
massicot2025almost, farokhi2016estimation, bacsar2024incentive}, and
communication without commitment, from cheap talk and signaling
\cite{crawford1982strategic, cho1987signaling} to the information-theoretic and
game-theoretic treatments of strategic sensing and signaling in the control
literature \cite{akyol2016information, saritas2017quadratic,
sayin2019hierarchical, saritas2020dynamic}. Dynamic information design extends
the commitment model over time \cite{ely2017beeps, honryo2018dynamic,
ball2023dynamic, doval2020sequential, farhadi2022dynamic, wu2022markov}, and
dynamic sender--receiver games without commitment are analyzed in
\cite{renault2013dynamic}; dynamic games with asymmetric information more
broadly are treated in \cite{basar1999dynamic, ouyang2017dynamic}. These models
either take the sender's credibility as given or derive it from equilibrium
inference by rational receivers. Our focus is instead on a credibility state
that is a mechanical public score, and on how the sender should manage it.

\paragraph{Reputation and career concerns.} A second strand concerns reputation
in the sense of a belief about a fixed sender type, where the sender's desire
to appear competent or unbiased distorts communication
\cite{ottaviani2006reputational, backus2020dont}. Such forces appear in media
markets \cite{gentzkow2006media}, firm relationships
\cite{prendergast1993theory}, politically charged advice
\cite{morris2001political}, and online markets \cite{ma2021reputation}. In our
model trust is not a belief about a type but a behavioral state that the sender
raises or lowers through its reporting choices.

\paragraph{Trust as a state variable.} Closest to our setting are models in
which trust is an explicit state that evolves with realized behavior and
governs the sender's future influence \cite{best2024persuasion,
mathevet2018reputation, barisaac2003reputation, manapat2013information,
zhu2018stability, zhu2019routing, zhu2022information, gurdal2014truth,
deori2024informativeness, vonsydow2019temporal}. In
\cite{manapat2013information} trust evolves across a population and can produce
recurrent trust cycles; in \cite{mathevet2018reputation} reputation is a belief
about a persistent type and the sender may collapse to myopic deception; and in
the routing applications of \cite{zhu2018stability, zhu2019routing} trust
evolves from realized outcomes and behavior converges to the planner's
recommendations. Reputational dynamics in trust games have also been studied
through learning and logit dynamics \cite{11107502}, and
\cite{best2024persuasion} shows that long-run reputational incentives generally
cannot fully substitute for commitment. In contrast to the type-based and
cycle-generating models, we treat trust as a scalar score that directly enters
the sender's signaling rule, and we show that in the fluid limit no policy
gains from build-and-exploit cycles.

\paragraph{Reputation for honesty and rating systems.} A fourth strand, and the
one closest to our motivation, studies reputation \emph{for honesty} and the
rating systems that carry it. In \cite{fudenberg2022honesty} a patient sender
announces intended actions and short-run receivers observe only whether past
promises were kept; \cite{jullien2014new} studies the credibility of a seller's
quality claims in an online market; \cite{ekmekci2011sustainable} designs
ratings that summarize a finite history so that reputations survive; and
\cite{board2013reputation} obtains work--shirk cycles for a reputation for
quality under Poisson news. These are equilibrium models with Bayesian
receivers and sender types. In ours the record itself is the object of
interest: it is a mechanical score, and it is censored, because only detected
misrepresentations are entered. Our finding that cycles do not pay should be
read against \cite{liu2014limited}, where limited records and fully rational
buyers do generate reputation bubbles and cycles. There, cycles arise from the
receivers' inference problem; here the sender's reward is affine in the lying
rate and switching between honesty and exploitation is costless, and both
features remove the benefit of a band. 

A related conclusion appears in
\cite{djawadi2018maintaining}, where a provider facing noisy ratings either
maintains a good reputation or never builds one; there the noise is a symmetric
rating error rather than one-sided censoring, and the result is obtained for a
finite reputation profile rather than a scalar score. Closest in spirit to our
receiver side are \cite{ko2025persuasion}, where short-lived receivers switch
between competing models by a Bayes-factor rule, and \cite{lukyanov2025mutual},
where inspections are public and the receiver also builds a reputation for
vigilance; relative to these, our sender has no commitment, controls a scalar
score directly, and faces a non-strategic buyer population. Finally,
\cite{luo2025marginal} studies reputation when short-run players observe the
long-run player's past actions but not her past signals, the platform-side
design of the record is studied in \cite{dellarocas2005reputation,
aperjis2010optimal}, and the empirical counterparts of our censoring mechanism
are \cite{dellarocas2008silence, cabral2010dynamics, filippas2022}.

\paragraph{Methods and applications.} Our trust update is a constant-step
stochastic approximation, and the concentration and fluid-limit arguments of
Section~\ref{sec:analysis} are the standard ones for such recursions
\cite{kushner2003stochastic, borkar2008stochastic}. The $(s,S)$ class is
borrowed from inventory control, where bands are optimal because of a fixed
ordering cost \cite{scarf1960optimality}; the absence of such a cost in our
setting is what makes the band collapse. Our detection mechanism connects to
communication with verifiable or detectable information
\cite{dziuda2018communication, titova2021persuasion}, and the application
relates to dynamic persuasion and information provision in markets
\cite{drakopoulos2021persuading, bonatti2023coordination, wu2021value}.

\subsection{Contributions}

\begin{itemize}

\item \emph{Censored score and the credibility premium.} Because only a
fraction $\gamma$ of misrepresentations enters the record, the score rests
above the sender's true honesty rate. We give the resting level in closed form,
show that it does not depend on the prior, and identify the lying rate at which
the gap is largest (Section~\ref{sec:premium}). Receiver behavior depends on
the two losses only through their ratio, the tolerance $M$, which we
micro-found as the median of a heterogeneous population
(Appendix~\ref{sec:receiver_foundation}).

\item \emph{What the sender is managing.} Against face-value receivers the
sender's gain from managing the score is the rent created by censoring. We make
this precise by comparing with receivers who know $\gamma$ and invert the
censoring: against them the best constant policy is a corner
(Proposition~\ref{prop:receiver_models}), and in the fluid limit no policy does
better (Proposition~\ref{prop:debias}).

\item \emph{A policy-free bound.} In the fluid limit, in which the score
adjusts slowly, no policy earns more than the best stationary value
$\max_r\bar\pi(r)$ (Theorem~\ref{thm:fluid}); the optimal threshold attains it
(Remark~\ref{rem:threshold_attains}). At a fixed adjustment speed $\rho$ every
policy is bounded by $\max_r\bar\pi(r)+O(\rho)$ (Corollary~\ref{cor:finite_rho}).

\item \emph{The three classes.} For the linear policy we show that the score
concentrates around a unique fixed point, derive the optimal parameter in
closed form, bound the error of the steady-state approximation by
$O(\sqrt\rho)$, and characterize the levels the class can reach
(Theorems~\ref{thm:linear_fixed_point_concentration}--\ref{thm:optimal_linear_alpha},
Proposition~\ref{prop:linear_regimes}). For the $(s,S)$ class we show, under a
regenerative-cycle approximation, that every nondegenerate band is weakly
dominated by a degenerate one (Theorem~\ref{thm:sS_collapse_threshold}). For
the threshold policy we obtain the optimal cutoff in closed form
(Theorem~\ref{thm:optimal_threshold}).

\item \emph{Equivalence and dominance.} All three classes earn the same
function $\bar\pi$ of the level they hold, so they differ only in which levels
they can reach. The threshold class strictly dominates the linear class exactly
when the tolerance exceeds an explicit cutoff $M_{\max}$
(Proposition~\ref{thm:equivalence}).

\item \emph{Receiver losses.} A larger lemon loss raises the optimal operating
level and makes the sender more truthful; a larger opportunity loss does the
opposite. Both effects act through $M$, and we characterize the limits of
extreme receiver preferences
(Propositions~\ref{prop:comparative_statics}--\ref{prop:loss_limits}).

\item \emph{Numerical study.} Simulations of the unapproximated chain confirm
the closed forms, measure how fast each policy reaches its operating level and how tightly it holds it, map the value of the credibility premium
over $(\gamma,M)$, and mark where the single-level conclusion stops. Bands pay once switching modes carries a fixed cost, and build-and-exploit cycles pay when demand is seasonal (Section~\ref{sec:numerical}).

\end{itemize}

A preliminary version of part of this work appeared in the
conference paper \cite{arcaklioglu2025parametric}. That paper focuses on the
symmetric-loss case ($c_l=1$, $c_o=0$) and analyzes the linear and threshold
policy classes separately; the invariant interval
\eqref{eq:invariant_interval}, the linear fixed point
\eqref{eq:linear_r_star}, the mean-square concentration result, and the
benchmark parameters of Proposition~\ref{prop:benchmark_linear} appear there.
The present paper extends the setup to the asymmetric loss structure and the tolerance $M$
(Section~\ref{sec:receiver}), the credibility premium of
Lemma~\ref{lem:rest_point}, the policy-free results of
Section~\ref{sec:fluid}, the $(s,S)$ class and its collapse
(Theorem~\ref{thm:sS_collapse_threshold}), the payoff equivalence and
dominance result (Proposition~\ref{thm:equivalence}), and the numerical study of
Section~\ref{sec:numerical}. 

The rest of the paper is organized as follows. Section~\ref{model} sets up the model, the trust dynamics, the receivers and
the three policy classes. Section~\ref{sec:analysis} characterizes the
credibility premium, introduces the benchmark receivers who correct for it,
proves the fluid-limit bound, and analyzes and compares the three classes. Section~\ref{sec:numerical}
reports the simulations, and Section~\ref{sec:conclusion} concludes. Proofs
not given in the text are in the Appendix.

\section{General Model} \label{model}

We consider repeated interactions between a long-run sender and a sequence of
myopic receivers, that is, receiver  $t$ participates only in period $t$, and therefore has no incentive to punish. At each time \(t\in \mathbb{N}\), the state
$s_t \in \mathcal{S}:=\{l,h\}$ is drawn i.i.d. with 
\begin{equation}
    \Prob(s_t=h)=p_h, \qquad \Prob(s_t=l)=1-p_h,
\end{equation}
where $p_h \in (0,1)$.  $s_t$ can refer to the quality of a product, high-quality $h$ or
lemon $l$. At each interaction between the sender and a receiver, the sender privately
observes \(s_t\) and sends a binary signal \(\theta_t\in \Theta :=\{L,H\}\) of lemon ($L$) or high ($H$) quality to the
receiver. The receiver then decides whether to purchase, with action
\(a_t\in \mathcal{A}:= \{0,1\}\), where \(a_t=1\) denotes purchase.

The receiver's payoff depends jointly on the product quality and the purchase
decision. Let \(c_l>0\) denote the \emph{lemon} loss from purchasing a
low-quality product, and let \(c_{o}\ge 0\) denote the \emph{opportunity} loss from not purchasing a high-quality product. The receiver's
one-period utility is
\begin{equation}
\begin{array}{c|cc}
 & s_t=h & s_t=l \\
\hline
a_t=1 & 1 & -c_l \\
a_t=0 & -c_o & 0
\end{array}
\label{eq:receiver_payoff}
\end{equation}
where the receiver's gain from buying a high quality product is normalized to
$1$. This normalization is without loss of generality because the receiver behavior depends on the
payoff table only through the ratio of the two differences $1+c_o$, the gain from
buying a high-quality product rather than skipping it, and $c_l$, the
loss from buying a low-quality product rather than skipping it. The sender
benefits whenever the purchase is made, so it seeks to induce purchase
by persuading the receiver that $s_t=h$. It is convenient to summarize these two differences by their ratio,
\begin{equation}
M:=\frac{1+c_o}{c_l},
\label{eq:M_def}
\end{equation}
which we refer to as the receiver's \emph{tolerance}. A larger $M$ means either
that buying a lemon is cheap or that missing a high-quality product is
expensive, so the receiver is more willing to act on a favorable signal. We show  that receiver behavior depends on $(c_l,c_o)$ only through $M$ (Appendix~\ref{sec:receiver_foundation}, summarized in Remark~\ref{rem:purchase_foundation}).

In many settings, the lemon loss  is far larger than the opportunity loss: that is, $ c_l \gg c_o \approx 0$.  $(c_l,c_o)$ can also be treated as the losses of a representative receiver; in
Appendix~\ref{sec:receiver_foundation} we allow these to vary across receivers, in which case
$M$ is the median tolerance in the population (Lemma~\ref{lem:microfoundation}).

\subsection{Trust Dynamics}
\label{sec:trust_dynamics}

We model trust as an endogenous scalar state \(r_t\in[0,1]\), interpreted as the
receiver's perceived probability that the sender reports a low signal when the state is low. This is a reputation score rather than a fully rational
posterior over the sender's strategy.

The score is a discounted empirical estimate of the sender’s honesty, in which recent outcomes carry more weight than earlier ones: an exponentially weighted moving average of the kind used by online feedback systems
\citep{xu2014onlineexclusive, tadelis2016}:
\begin{equation}
r_{t+1}=(1-\rho)r_t+\rho \lambda_t,
\label{eq:trust_update}
\end{equation}
where \(\rho\in(0,1)\) is the trust adjustment speed. The parameter
\(\lambda_t\in\{0,1,r_t\}\) depends on the realized state, the signal, and
whether deception is detected:
\begin{equation}
\lambda_t=
\begin{cases}
1, & s_t=l,\ \theta_t=L \quad\text{(honest low signal)},\\[2pt]
0, & s_t=l,\ \theta_t=H,\ d_t=1 \quad\text{(detected misrepresentation)},\\[2pt]
r_t, &  s_t=h,\ \text{or } s_t=l,\ \theta_t=H,\ d_t=0\text{ (high state or undetected misrepresentation)  },
\end{cases}
\label{eq:lambda_update}
\end{equation}
where $d_t\sim\mathrm{Bernoulli}(\gamma)$ is drawn i.i.d.\ and independently of
$(s_t,\theta_t,r_t)$ with $d_t=1$ indicating that a misrepresentation in
$t$ is detected, and $\gamma\in(0,1]$ the probability of detection.

Since $\lambda_t\in[0,1]$ and \eqref{eq:trust_update} is a convex combination of
$r_t$ and $\lambda_t$, an initial condition $r_0\in(0,1)$ implies
$r_t\in(0,1)$ for every $t\ge 0$,
and $r_0\in[0,1]$ implies $r_t\in[0,1]$; we use this fact in
Section~\ref{sec:analysis}. Where the score comes to rest under a given
reporting policy, and the fact that it overstates the sender's honesty when
detection is imperfect, are the subject of Section~\ref{sec:premium}.

\subsection{Information, Strategies, and Objective}
\label{sec:strategies_objective}

The trust level $r_t$ is a public statistic maintained by the platform. At the
start of period $t$ the sender privately observes the state $s_t$ and publicly
observes $r_t$; receiver $t$ observes $(\theta_t,r_t)$ but not $s_t$, and acts
only in period $t$. Because $r_t$ summarizes past outcomes and no receiver
participates twice, we take $(\theta_t,r_t)$ to be the receiver’s payoff-relevant information. Whether receivers additionally observe the history $\{(\theta_k,a_k)\}_{k<t}$ is immaterial for our analysis, because the purchase rule we adopt in Section~\ref{sec:receiver} conditions on the past only through $r_t$.

\begin{assumption}[Truthful high signals]
\label{as:truthful_high}
$\Prob(\theta_t=H\mid s_t=h)=1$ for all $t$.
\end{assumption}

By \eqref{eq:lambda_update}, no signal in the high state affects the score. Given any receiver decision-making rule nondecreasing in the posterior, a high signal then increases the purchase probability at no cost, so the signal $\theta_t=H$ dominates when $s_t=h$. Under the assumption, $L$ is \emph{off path} in the high state, so strategic
behavior arises only in the low state. A \emph{policy} $\sigma$
therefore specifies, in each period $t$, a lying rate
\begin{equation}
u_t:=\Prob\bigl(\theta_t=H\mid s_t=l,\ \mathcal F_t\bigr)\in[0,1]
\label{eq:policy_def}
\end{equation}
as a function of the history $\mathcal F_t$, generated by $r_0$ and
$\{(s_k,\theta_k,d_k,a_k)\}_{k<t}$; in particular, the score $r_t$ is
determined by $\mathcal F_t$. The policy is \emph{stationary Markov} if
$u_t=\sigma(r_t)$ for a measurable map $\sigma:[0,1]\to[0,1]$, in which case
$\{r_t\}$ is a Markov chain on $[0,1]$. Any policy $\sigma$ together with
$r_0\in[0,1]$ induces a law $\Prob^{\sigma}_{r_0}$ on histories via
\eqref{eq:trust_update}--\eqref{eq:lambda_update}. The sender's objective is to maximize
the long-run average sales rate
\begin{equation}
\mathcal J(\sigma,r_0)
:=
\liminf_{T\to\infty}
\frac{1}{T}\,
\E^{\sigma}_{r_0}\!\left[\sum_{t=0}^{T-1} a_t\right].
\label{eq:sender_objective}
\end{equation}

We do not compute $\sup_\sigma\mathcal J(\sigma,r_0)$ in closed form. Instead,
Section~\ref{sec:analysis} bounds it for every policy and optimizes it over
three parametric classes, using approximations whose error it controls; none
of these results requires the invariant distribution of the score to be
unique.

\subsection{Receiver Beliefs and Strategy}
\label{sec:receiver}

Let $\mu:=\Prob(s_t=h\mid H,r_t)$ denote the receiver's posterior belief that the
product is high quality after observing a high signal,
\begin{equation}
\mu
=
\frac{\Prob(H|h,r_t) \Prob(h) }{\Prob(H|r_t)}
=
\frac{p_h}{1-(1-p_h)r_t}.
\label{eq:posterior_H}
\end{equation}
Here the receiver uses $r_t$ as the perceived probability of a truthful low
signal, so that the implied misrepresentation rate is $1-r_t$. Since the sender is truthful in the high state, \(\Prob(L\mid h)=0\), observing \(L\) directly reveals that $s_t=l$, and the receiver does not purchase
after \(L\). Throughout the analysis we work with a smooth purchase rule in which the share of receivers who buy after a high signal is
 \begin{equation}
 q(r):=\Prob(a_t=1\mid \theta_t=H,\,r_t=r)=\frac{\mu(1+c_{o})}{\mu(1+c_o)+(1-\mu)c_l}=\frac{p_h(1+c_o)}
{p_h(1+c_o)+(1-p_h)(1-r)c_l}.
 \label{eq:probabilistic_purchase}
 \end{equation}

\begin{remark}[Foundation of the purchase rule]
\label{rem:purchase_foundation}
The purchase rule \eqref{eq:probabilistic_purchase} is not a behavioral
departure from expected-utility maximization. A single receiver buys after a high
signal if and only if $\mu\ge 1/(1+M)$, so its losses enter only through the
tolerance $M$ (Proposition~\ref{prop:receiver_threshold}). If tolerances vary
across receivers according to a log-logistic law with median $M$, the share of
receivers who buy is exactly \eqref{eq:probabilistic_purchase}
(Lemma~\ref{lem:microfoundation}). Equivalently, the odds of purchase equal the
odds of high quality scaled by $M$: $q/(1-q)=M\mu/(1-\mu)$. Formal statements and
proofs are in Appendix~\ref{sec:receiver_foundation}, along with a discussion of
which results depend on this distributional choice.
\end{remark}

Receivers who read $r_t$ at face value are the population we analyze.
Section~\ref{sec:premium} introduces a benchmark population that knows the
detection probability and corrects for the censoring of the score, against
which the value of the face-value reading is measured.

\subsection{Sender Policies}\label{strategies}

We consider three trust-dependent policy classes that capture
different ways in which the sender may manage credibility over time. Each class below specifies $\sigma$ up to one or two scalar parameters; for readability we
abbreviate $\sigma(r_t)$ by $\Prob(H\mid l)$, leaving the dependence on $r_t$
implicit.

\subsubsection{Linear Policy}

Under the linear policy, the sender's probability of misrepresenting a
low-quality product increases smoothly with the current trust level:
\begin{equation}
\Prob(H\mid l)=\frac{r_t}{\alpha},
\qquad \alpha\ge 1.
\label{eq:linear_policy}
\end{equation}
Policy parameter $\alpha$ controls the sender's aggressiveness. A smaller
$\alpha$ implies a higher probability of deception, while a larger $\alpha$ makes
the sender more conservative. For a fixed $\alpha$, when trust is low, the sender
is less willing to risk further credibility loss, whereas if trust is high, the
sender exploits that credibility more frequently. The restriction $\alpha\ge 1$
guarantees $r_t/\alpha\in[0,1]$.

\subsubsection{$(s,S)$ Policy}

Under the $(s,S)$ policy, the sender follows a band-based rule, similar to
inventory control problems. The lower threshold $s$ represents the level below
which credibility is considered too low to exploit, while the upper threshold $S$
is the level at which the sender is willing to begin exploiting
credibility. When trust is below $s$, the sender truthfully reports low quality.
Once trust reaches $S$, the sender switches to deception in the low state and exploits until  trust falls below $s$, i.e.,
\begin{equation}
\Prob(H\mid l)=\mathbf 1\{m_t=\mathrm E\},
\qquad
m_{t+1}=
\begin{cases}
\mathrm E, & r_{t+1}\ge S,\\
\mathrm B, & r_{t+1}<s,\\
m_t, & s\le r_{t+1}<S,
\end{cases}
\label{eq:sS_policy}
\end{equation}
where $0\le s\le S \le1$, $m_t\in\{\mathrm B,\mathrm E\}$ denotes the current mode,
building ($\mathrm B$) or exploiting ($\mathrm E$), and $m_0$ is determined from $r_0$.
The middle region $s\le r_t<S$ introduces a hysteresis, where the sender does not
switch modes immediately inside the band, but instead keeps the mode with which
the process entered the band. If trust enters the band from below, the sender
remains truthful until trust reaches $S$, and conversely. Thus, the $(s,S)$
policy explicitly captures a build-and-exploit pattern where the sender
alternates between periods of rebuilding and exploitation of trust.

\subsubsection{Threshold-based Policy}

Finally, we consider a threshold-based policy, which is the special case of
the $(s,S)$ policy where the two thresholds coincide, $s=S=\tau$, so that the band is empty and the mode $m_t$ is determined by $r_t$
alone. In this case, the sender uses a single cutoff:
\begin{equation}
\Prob(H\mid l)=
\begin{cases}
0, & r_t<\tau,\\
1, & r_t\ge \tau,
\end{cases}
\label{eq:threshold_policy}
\end{equation}
where $\tau \in [0,1]$. When trust is below the threshold, the sender reports low-quality products
truthfully to rebuild credibility and conversely.

\section{Analysis Under Asymmetric Loss}
\label{sec:analysis}

We now turn to the sender's credibility-management problem.
Section~\ref{sec:premium} first records where the score comes to rest and
shows that, because only detected misrepresentations are entered, it rests
above the sender's true honesty rate; it also introduces the benchmark
receivers who correct for this. Section~\ref{sec:fluid} then treats all
policies at once: in the fluid limit no policy earns more than operating at a
single credibility level, and the sender's entire gain from doing so is the
rent created by censored detection. The subsections that follow analyze the
three classes of Section~\ref{strategies}, which differ in which resting
levels of Lemma~\ref{lem:rest_point} they can reach, derive how the
receiver's losses shape their optimal parameters, and compare them in
Section~\ref{sec:equivalence}. The case of interest is $\gamma<1$: under perfect detection the trust equals the sender’s true honesty rate by Lemma~\ref{lem:rest_point}, so the credibility premium vanishes. We state the $\gamma=1$ cases where they are informative and assume $\gamma<1$ from Section~\ref{sec:linear} on.

\subsection{Resting Score, Credibility Premium, and De-biasing Receivers}
\label{sec:premium}

Detection imperfection has a consequence that is central to the
paper, and we record it before turning to the sender's problem. Because only a
fraction $\gamma$ of misrepresentations is recorded, the score $r_t$ does
not settle at the sender's true honesty rate. The following lemma computes the
resting level of trust induced by a constant misrepresentation rate
and quantifies the resulting gap.

\begin{lemma}[Resting trust and the credibility premium]
\label{lem:rest_point}
Fix $\gamma\in(0,1]$, $\rho\in(0,1)$, and suppose the sender
misrepresents low quality at a constant rate
$u:=\Prob(\theta_t=H\mid s_t=l)\in[0,1]$. Then the trust dynamics
\eqref{eq:trust_update}--\eqref{eq:lambda_update} have a unique zero-drift point
in $[0,1]$, given by
\begin{equation}
r^*(u)=\frac{1-u}{1-u+\gamma u},
\label{eq:rest_point}
\end{equation}
which is independent of $p_h$. Moreover $r^*(u)\ge 1-u$, with equality
if and only if $\gamma=1$ or $u\in\{0,1\}$.
\end{lemma}

\begin{proof}
By \eqref{eq:trust_update},
$\E[r_{t+1}-r_t\mid r_t=r]=\rho\bigl(\E[\lambda_t\mid r_t=r]-r\bigr)$,
so the zero-drift condition is $\E[\lambda_t\mid r_t=r]=r$. From
\eqref{eq:lambda_update},
\begin{equation}
\E[\lambda_t\mid r_t=r]
=p_h r+(1-p_h)\bigl[(1-u)+u(1-\gamma)r\bigr].
\end{equation}
Subtracting $p_h r$ from both sides and dividing by
$1-p_h>0$ gives $(1-u)+u(1-\gamma)r=r$, which gives \eqref{eq:rest_point}. Uniqueness
follows since the equation is affine in $r$ with slope $(1-\gamma)u<1$. As
$(1-u)+\gamma u\le 1$ for $\gamma\le1$, we get $r^*(u)\ge 1-u$, with equality iff
$(1-\gamma)u=0$ or $u=1$.
\end{proof}

Two features of \eqref{eq:rest_point} are worth emphasizing. First, the
prior $p_h$ cancels: it governs how \emph{fast} trust moves, but not where it
comes to rest. Second,
the resting level depends on the sender's behavior only through the single rate
$u$, so \eqref{eq:rest_point} applies to any policy that induces a constant
misrepresentation rate. We use it in Section~\ref{sec:linear} to obtain the
fixed point of the linear policy as an immediate corollary.

Trust moves only when $\lambda_t\in\{0,1\}$, call such periods \emph{informative}.
Informative periods arise either from an honest low signal, at rate
$(1-p_h)(1-u)$, or from a \emph{detected} misrepresentation, at rate
$(1-p_h)u\gamma$.  Trust
therefore records behavior in the low state only: periods in which the product is
high quality carry no information about how the sender behaves when it is not. Conditional on a period being informative, the probability that it
is of the honest type is exactly $r^*(u)$ in \eqref{eq:rest_point}. Thus, $r_t$ is
a discounted empirical honesty rate computed over the subsample of sender
behavior that leaves observable evidence.  Undetected lies are  indistinguishable from periods in
which no low-quality unit arose.

Consequently, receivers who take $r_t$ at face value---as they do in
\eqref{eq:posterior_H}---overstate the sender's honesty by the factor
$\bigl(1-(1-\gamma)u\bigr)^{-1}\in[1,1/\gamma]$, and understate its lying rate by the
factor $\gamma/\bigl(1-(1-\gamma)u\bigr)$, which tends to $\gamma$ as $u\to0$: when
misrepresentation is rare, only a fraction $\gamma$ of it is ever priced in. We call the wedge $r^*(u)-(1-u)\ge0$ the \emph{credibility premium}. It is
created by the censoring of the record rather than by the receivers, and it
vanishes under perfect detection, since $\gamma=1$ implies $r^*(u)=1-u$; the
receivers matter only because they read the record at face value. The
premium is largest at an interior lying rate: the wedge
$r^*(u)-(1-u)$ is maximized at $u^\sharp=(1+\sqrt \gamma)^{-1}$, where it equals
$(1-\sqrt \gamma)/(1+\sqrt \gamma)$ and where $r^*(u^\sharp)=u^\sharp$. The rents that the sender extracts in the rest of this section are the value of this premium.

\begin{remark}[Interpreting detection]
\label{rem:detection}
In \eqref{eq:lambda_update} the detection indicator $d_t$ is independent of the
purchase decision $a_t$, so a misrepresentation can be recorded even in periods
without a sale. We read $\gamma$ as covering listing audits, third-party
inspection, mystery shopping and disclosure through other channels, all of which
operate whether or not this particular receiver bought. If instead detection
required a purchase---as returns, reviews and warranty claims do---the detection
probability would be $\gamma q(r_t)$ and Lemma~\ref{lem:rest_point} would become
the fixed-point equation $r^*=(1-u)/\bigl(1-(1-\gamma q(r^*))u\bigr)$, which
couples the premium to receiver tolerance: more tolerant buyers buy more, detect
more, and shrink the premium. That variant is economically interesting, and we
report it as a robustness check in Section~\ref{sec:numerical}; it does not
change the structure of the sender's problem, since the reward and the drift
remain affine in $u$ for fixed $r$.
\end{remark}

\paragraph{Face-value and de-biasing receivers.}
The posterior \eqref{eq:posterior_H} reads $1-r_t$ as the sender's lying
rate. By Lemma~\ref{lem:rest_point}, this understates the rate that produced the
score: a sender who lies at a constant rate $u\in(0,1)$ holds the score at
$r^*(u)\geq 1-u$. As a benchmark, we also consider \emph{de-biasing} receivers, who
know the detection probability $\gamma$ and correct for the censoring. Since
$u\mapsto r^*(u)$ is strictly decreasing, it can be inverted, and from the score
$r$ such a receiver infers the lying rate
\begin{equation}
    \hat u(r):=\frac{1-r}{D(r)},
\qquad\text{so that}\qquad
\hat u\bigl(r^*(u)\bigr)=u ,\label{eq:u_hat}
\end{equation}
where $D(r)=1-(1-\gamma)r$, which satisfies $D\ge\gamma$ on $[0,1]$. By Lemma~\ref{lem:rest_point}, this is also the only constant lying rate that holds the score at $r$. A de-biasing receiver therefore forms the posterior
\begin{equation}
\mu^c(r):=\frac{p_h}{p_h+(1-p_h)\,\hat u(r)}
\label{eq:consistent_posterior}
\end{equation}
in place of \eqref{eq:posterior_H}. When the score is at the resting level of a
constant lying rate $u$, the estimate is exact, and
\eqref{eq:consistent_posterior} is the Bayesian posterior under the sender's
true lying rate. This is the sense in which de-biasing receivers know the
sender's behavior. They have the same losses as face-value receivers, so
replacing $\mu$ by $\mu^c(r)$ in \eqref{eq:probabilistic_purchase} gives their
purchase share
\begin{equation}
q^c(r):=\frac{p_hM}{p_hM+(1-p_h)\,\hat u(r)} .
\label{eq:q_c}
\end{equation}

\begin{proposition}[Face-value versus de-biasing receivers]
\label{prop:receiver_models}
\begin{enumerate}
\item[(i)] At every score $r\in[0,1]$, the face-value posterior
\eqref{eq:posterior_H} weakly exceeds the de-biasing posterior
\eqref{eq:consistent_posterior}, with equality if and only if $\gamma=1$ or
$r\in\{0,1\}$.
\item[(ii)] Suppose the sender misrepresents low quality at a constant rate
$u\in[0,1]$, and the score is at its resting level $r^*(u)$ of
Lemma~\ref{lem:rest_point}. Then de-biasing receivers hold the Bayesian
posterior $\mu^c\bigl(r^*(u)\bigr)=p_h/\bigl(p_h+(1-p_h)u\bigr)$, and the
sender's sales rate against them is
\begin{equation}
V(u):=x\,q^c\bigl(r^*(u)\bigr)=\frac{x\,p_hM}{p_hM+x-p_h},
\qquad x:=p_h+(1-p_h)u .
\label{eq:debias_value}
\end{equation}
$V$ is increasing in $u$ if $M>1$, decreasing if $M<1$, and constant if $M=1$.
Hence
\begin{equation}
\max_{u\in[0,1]}V(u)=\max\{p_h,\;q(0)\},
\label{eq:debias_max}
\end{equation}
attained by always-honest behavior ($u=0$) when $M\le1$ and by always-deceptive
behavior ($u=1$) when $M\ge1$.
\end{enumerate}
\end{proposition}

\begin{proof}
(i) The face-value posterior is $p_h/\bigl(p_h+(1-p_h)(1-r)\bigr)$, so it
is at least $\mu^c(r)$ if and only if $1-r\le\hat u(r)$. For $r<1$, this is
equivalent to $1-(1-\gamma)r\le1$, which always holds, with equality if and
only if $(1-\gamma)r=0$. At $r=1$, both rates are zero.

(ii) Since $\hat u\bigl(r^*(u)\bigr)=u$, the posterior is as stated and
$q^c\bigl(r^*(u)\bigr)=p_hM/\bigl(p_hM+(1-p_h)u\bigr)$. By
Assumption~\ref{as:truthful_high}, a high signal is sent with probability $x$.
Since $(1-p_h)u=x-p_h$, the product is \eqref{eq:debias_value}.
Differentiating in $x$ gives $V'(x)=p_h^2M(M-1)/(p_hM+x-p_h)^2$, which has the
sign of $M-1$. Finally, $V(0)=p_h$ and $V(1)=p_hM/(p_hM+1-p_h)=q(0)$.
\end{proof}

Part~(i) is the receiver-side reading of the credibility premium: at every
interior score, face-value receivers are more optimistic than receivers who
correct for censoring. Part~(ii) shows that without this optimism the sender has
nothing to manage, since the best constant policy is a corner.
Proposition~\ref{prop:debias} in Section~\ref{sec:fluid} extends (ii) to every
policy.

\subsection{Fluid Limit, Global Optimality, and the Value of the Premium}
\label{sec:fluid}

Before analyzing the three classes of Section~\ref{strategies} one at a
time, we ask a question that does not depend on the class: how much can
\emph{any} policy earn when the score adjusts slowly? The answer rests on a
deterministic approximation of the trust dynamics, which we now set up.

In each period the score moves by $r_{t+1}-r_t=\rho(\lambda_t-r_t)$, which is
at most $\rho$ in absolute value. Since $s_t$, $d_t$ and the period-$t$ receiver are
drawn independently of $\mathcal F_t$, a policy with lying rate $u_t$
satisfies $\E[r_{t+1}-r_t\mid\mathcal F_t]=\rho\,b(r_t,u_t)$ and
$\E[a_t\mid\mathcal F_t]=\pi(r_t,u_t)$, where
\begin{equation}
b(r,u):=(1-p_h)\bigl[(1-r)-u\,(1-(1-\gamma)r)\bigr],
\qquad
\pi(r,u):=\bigl(p_h+(1-p_h)u\bigr)\,q(r)
\label{eq:drift_reward}
\end{equation}
are the expected change of the score, measured in units of $\rho$, and the
expected sale in a period that starts at score $r$ with lying rate $u$.
If time is measured in units of $1/\rho$ periods, that is, if one follows the
rescaled path $r^\rho(t):=r_{\lfloor t/\rho\rfloor}$, then over one unit of
time the drift accumulates to a change of order one, while the random
fluctuations, being $1/\rho$ roughly independent terms of size $\rho$,
accumulate only to order $\sqrt\rho$. As $\rho\to0$ the rescaled path
therefore follows the ordinary differential equation
\begin{equation}
\dot r(t)=b\bigl(r(t),u(t)\bigr),
\label{eq:fluid_model}
\end{equation}
where $u(t)\in[0,1]$ is the lying rate in use at time $t$. We refer to
\eqref{eq:fluid_model} as the \emph{fluid model}. Every policy, whether
Markov, history-dependent or randomized, generates a lying-rate path
$u(\cdot)$ in the fluid model, and we allow any measurable path with values
in $[0,1]$: this is the natural regularity under which \eqref{eq:fluid_model}
has a solution and the time-averages below are well defined.

Two features of \eqref{eq:drift_reward} drive everything that follows. First,
both $b$ and $\pi$ are affine in $u$. Second, for each $r\in[0,1]$ the drift
vanishes at exactly one lying rate: $\hat u(r)$ of \eqref{eq:u_hat}.  This inverts the map $u\mapsto r^*(u)$ of Lemma~\ref{lem:rest_point}: lying
at rate $\hat u(r)$ is the only way to keep the score at $r$. It is the same
function the de-biasing receivers of Section~\ref{sec:premium} use to read
the score: the rate that holds the score at $r$ is the rate they infer from
it. We define the \emph{stationary value} of operating at credibility level
$r$ as
\begin{equation}
\bar\pi(r):=\pi\bigl(r,\hat u(r)\bigr)
=\bigl(p_h+(1-p_h)\hat u(r)\bigr)\,q(r)
=\frac{1-(1-p_h\gamma)r}{D(r)}\,q(r),
\label{eq:common_value}
\end{equation}
the per-period sales rate earned when the score is held fixed at $r$. It
depends on $r$ alone, not on a policy, on $\rho$, or on $r_0$, and is
therefore distinct from the long-run payoff $\mathcal J(\sigma,r_0)$ of a
policy in \eqref{eq:sender_objective}. Since $\hat u(1)=0$ and
$\hat u(0)=1$, the endpoints $\bar\pi(1)=p_h$ and $\bar\pi(0)=q(0)$ are the
values of always-honest and always-deceptive behavior. The theorem below
shows that, in the fluid model, no policy earns more than the best stationary
value.

\begin{theorem}[Fluid-limit upper bound]
\label{thm:fluid}
For every measurable lying-rate path $u:[0,\infty)\to[0,1]$, every
initial condition $r(0)\in[0,1]$ and every $T>0$, the solution of the fluid
model \eqref{eq:fluid_model} satisfies
\begin{equation}
\frac1T\int_0^T\pi\bigl(r(t),u(t)\bigr)\,dt
\;\le\;
\max_{r\in[0,1]}\bar\pi(r)+\frac{2}{\gamma T}.
\label{eq:fluid_bound}
\end{equation}
\end{theorem}

\begin{proof}
 Let $\chi(r):=-q(r)/D(r)$ and
$K(r):=\int_0^r\chi(v)\,dv$. We first show that for every $r\in[0,1]$ and
every $u\in[0,1]$,
\begin{equation}
\pi(r,u)=\bar\pi(r)+\chi(r)\,b(r,u).
\label{eq:affine_identity}
\end{equation}
Fix $r$. Both $\pi(r,\cdot)$ and $b(r,\cdot)$ are affine in $u$, with slopes
$(1-p_h)q(r)$ and $-(1-p_h)D(r)$, and $b(r,\hat u(r))=0$ by \eqref{eq:u_hat}.
Expanding both around $\hat u(r)$,
\begin{equation}
    \pi(r,u)=\bar\pi(r)+(1-p_h)q(r)\bigl(u-\hat u(r)\bigr),
\qquad
b(r,u)=-(1-p_h)D(r)\bigl(u-\hat u(r)\bigr),
\end{equation}
and eliminating $u-\hat u(r)$ gives
$\pi(r,u)=\bar\pi(r)-\bigl(q(r)/D(r)\bigr)\,b(r,u)$, which is
\eqref{eq:affine_identity}.

Because $b$ is affine, and hence Lipschitz, in
$r$ uniformly in $u$, the fluid model \eqref{eq:fluid_model} has a unique
absolutely continuous solution for every measurable $u(\cdot)$, and the
solution remains in $[0,1]$ because $b(0,u)\ge0\ge b(1,u)$. Since $D\ge\gamma$
and $0\le q\le1$ on $[0,1]$, we have $|\chi|\le1/\gamma$ and hence
$|K|\le1/\gamma$. Along the solution, \eqref{eq:affine_identity} and the chain
rule give
\begin{equation}
    \pi\bigl(r(t),u(t)\bigr)=\bar\pi\bigl(r(t)\bigr)+\chi\bigl(r(t)\bigr)\,\dot r(t)
=\bar\pi\bigl(r(t)\bigr)+\tfrac{d}{dt}K\bigl(r(t)\bigr),
\end{equation}
so
\begin{equation}
    \frac1T\int_0^T\pi\bigl(r(t),u(t)\bigr)\,dt
=\frac1T\int_0^T\bar\pi\bigl(r(t)\bigr)\,dt+\frac{K(r(T))-K(r(0))}{T}
\le\max_r\bar\pi(r)+\frac{2}{\gamma T}.
\end{equation}
\end{proof}

\begin{remark}[Credibility as a stock]
\label{rem:no_cycles}
Every policy---Markov, history-dependent, randomized or cyclic---generates a
lying-rate path in the fluid model, so letting $T\to\infty$ in
\eqref{eq:fluid_bound}, none earns more than $\max_r\bar\pi(r)$ per unit of
time. The identity \eqref{eq:affine_identity} explains why: credibility
behaves like a stock with a state-dependent but path-independent shadow price
$\chi$, so a period's reward is the value of operating at $r$ plus the value
of credibility accumulated or spent, and the second term telescopes. A
build-and-exploit cycle returns the stock to where it started and therefore
buys nothing, while running the stock down once yields at most the bounded
term $2/(\gamma T)$.
\end{remark}

\begin{lemma}
\label{lem:pibar_shape}
For $r\in[0,1]$,
\begin{equation}
\bar\pi'(r)=\frac{(1-p_h)\,q(r)^2}{p_hM\,D(r)^2}\,F(r,M),
\label{eq:pibar_prime}
\end{equation}
where $F$ is a quadratic function in $r$ (see   \eqref{F_long_form} for its explicit form) strictly decreasing on $[0,1]$, with $F(0,M)=1-\gamma(1-p_h)-p_h\gamma M$ and
$F(1,M)=p_h\gamma(\gamma-M)$. Hence $\bar\pi$ has a unique maximizer $r_{\bar\pi}$ on
$[0,1]$: it is strictly increasing on $[0,r_{\bar\pi}]$ and strictly decreasing on
$[r_{\bar\pi},1]$. If $F(1,M)\ge0$, i.e.\ $M\le\gamma$, then $r_{\bar\pi}=1$; if
$F(0,M)\le0$, then $r_{\bar\pi}=0$; otherwise $r_{\bar\pi}\in(0,1)$ is the smaller
root of $F$,
\begin{equation}
r_{\bar\pi}=\frac{1}{1-p_h\gamma}\left[1-\sqrt{\frac{p_h\gamma\bigl[(1-p_h\gamma)M-\gamma(1-p_h)\bigr]}{1-\gamma}}\;\right].
\label{eq:r_pibar}
\end{equation}
\end{lemma}

Lemma~\ref{lem:pibar_shape} locates the maximum in the bound of
Theorem~\ref{thm:fluid}; its proof is in the Appendix. The class-specific optima
derived below are this same point. The bound \eqref{eq:fluid_bound} is a statement about the fluid limit. The
same accounting identity \eqref{eq:affine_identity} also bounds every policy
on the unapproximated chain at a fixed $\rho$: Corollary~\ref{cor:finite_rho}
in Appendix~\ref{app:finite_rho} shows that
$\mathcal J(\sigma,r_0)\le\max_r\bar\pi(r)+\tfrac12 C_\chi\rho$ for every policy
$\sigma$ and every $r_0$, with an explicit constant $C_\chi$ that grows as
detection becomes rare or receivers become intolerant. The bound applies in
particular to the $(s,S)$ class, which is not stationary in $r_t$ alone.

We close this subsection by extending
Proposition~\ref{prop:receiver_models}(ii) from constant lying rates to every
policy. Against the de-biasing receivers of Section~\ref{sec:premium}, the
reward at score $r$ and lying rate $u$ is
$\tilde\pi(r,u):=\bigl(p_h+(1-p_h)u\bigr)\,q^c(r)$, with $q^c$ from
\eqref{eq:q_c}. The drift $b(r,u)$ is unchanged, because the score does not
depend on how receivers read it.

\begin{proposition}[Value of the credibility premium]
\label{prop:debias}
In the fluid model \eqref{eq:fluid_model}, no policy earns more than
$\max\{\bar\pi(0),\bar\pi(1)\}$ per unit of time in the long-run against de-biasing receivers.
Hence the value of credibility management against face-value receivers,
\begin{equation}
\max_{r\in[0,1]}\bar\pi(r)-\max\{\bar\pi(0),\bar\pi(1)\},
\label{eq:rent}
\end{equation}
is entirely the rent created by imperfect detection, and it is zero when
$\gamma=1$.
\end{proposition}

\begin{proof}
Set $\tilde{\bar\pi}(r):=\tilde\pi\bigl(r,\hat u(r)\bigr)$ and
$\tilde\chi(r):=-q^c(r)/D(r)$. The two sides of
\begin{equation}
\tilde\pi(r,u)=\tilde{\bar\pi}(r)+\tilde\chi(r)\,b(r,u)
\label{eq:affine_identity_debias}
\end{equation}
are affine in $u$ and agree at $u=\hat u(r)$, where $b$ vanishes by
\eqref{eq:u_hat}. They also have the same coefficient of $u$, namely
$(1-p_h)q^c(r)$. Since $|\tilde\chi|\le1/\gamma$, the argument of
Theorem~\ref{thm:fluid} applies verbatim and bounds every lying-rate path by
$\max_r\tilde{\bar\pi}(r)+2/(\gamma T)$.

By \eqref{eq:q_c}, $\tilde{\bar\pi}(r)=\bigl(p_h+(1-p_h)\hat u(r)\bigr)q^c(r)
=V\bigl(\hat u(r)\bigr)$, with $V$ from \eqref{eq:debias_value}. Since $\hat u$
maps $[0,1]$ onto $[0,1]$, $\max_r\tilde{\bar\pi}(r)=\max_uV(u)=\max\{p_h,q(0)\}$
by \eqref{eq:debias_max}. By the endpoint values after \eqref{eq:common_value},
this equals $\max\{\bar\pi(1),\bar\pi(0)\}$.

If $\gamma=1$, then $\hat u(r)=1-r$, so $q^c=q$ and $\bar\pi(r)=V(1-r)$. The two
receiver models then coincide, $\bar\pi$ is monotone, and \eqref{eq:rent} is
zero.
\end{proof}

Against de-biasing receivers, the score conveys the sender's true behavior,
so the sender's problem is static in the fluid limit and its optimum is a
corner. Every rent in Section~\ref{sec:analysis} is therefore the value of the
censoring documented in Lemma~\ref{lem:rest_point}. At the default parameters
($p_h=\gamma=0.3$, $M=1$), the sender earns $0.4225$ against face-value
receivers and $0.300$ against de-biasing ones, a gain of $41\%$.
Section~\ref{sec:num_premium} traces this gap as $\gamma$ varies.

\subsection{Linear Policy}\label{sec:linear}

Our first class-specific result concerns the linear policy. We find that, for any
fixed policy parameter $\alpha$, the induced trust process concentrates around a
steady-state level $r^*_{lin}(\alpha) \in(0,1)$. This limiting trust level is
independent of the receiver's loss parameters because it is governed by the
trust-update rule. It depends on the sender's reporting behavior and the
detection probability.
By Lemma~\ref{lem:rest_point} it is also independent of the prior $p_h$.
Additionally, we characterize the optimal policy parameter by utilizing the
long-run trust level generated by the policy.

\begin{theorem}[Fixed point and concentration]
\label{thm:linear_fixed_point_concentration}
Suppose the sender follows the linear policy \eqref{eq:linear_policy} with
$\alpha\ge1$, let $p_h\in(0,1)$, $\gamma \in(0,1)$, $\rho\in(0,\tfrac12]$, initialized with $r_0\in(0,1)$. Then:
\begin{enumerate}
\item[(i)] The conditional mean trust dynamics have a unique fixed point in
$(0,1)$, given by
\begin{equation}
r^*_{lin}(\alpha)
:=
\frac{
\alpha+1-\sqrt{(\alpha+1)^2-4\alpha(1-\gamma)}
}{
2(1-\gamma)
}.
\label{eq:linear_r_star}
\end{equation}
\item[(ii)] With $\kappa:=(1-p_h)\bigl[1-(1-2\gamma)/\alpha\bigr]>0$, for every
$\varepsilon>0$,
\begin{equation}
\limsup_{t\to\infty}
\Prob\left(|r_t-r^*_{lin}(\alpha)|\ge \varepsilon\right)
\le
\frac{\rho}{4\kappa(2-\rho\kappa)\varepsilon^2}.
\label{eq:linear_prob_concentration}
\end{equation}
\end{enumerate}
Thus, the long-run deviation of $r_t$ from $r^*_{lin}(\alpha)$ is
$O(\sqrt\rho)$, and it vanishes as $\rho\to0$.
\end{theorem}

\begin{proof}
\textbf{(i)} Under \eqref{eq:linear_policy} the misrepresentation rate at trust
level $r$ is $u=r/\alpha$. Substituting into the zero-drift relation
\eqref{eq:rest_point} of Lemma~\ref{lem:rest_point} and clearing denominators
gives
\begin{equation}
(1-\gamma)r^2-(\alpha+1)r+\alpha=0 .
\label{eq:linear_fp_quadratic}
\end{equation}
The larger root of \eqref{eq:linear_fp_quadratic} exceeds one, while the smaller
root belongs to $(0,1)$; the latter is \eqref{eq:linear_r_star}.

\textbf{(ii)} Let $g(r):=\E[\lambda_t\mid r_t=r]$ and
$f(r):=\E[r_{t+1}\mid r_t=r]=(1-\rho)r+\rho g(r)$, so that $f$ and $g$ share the
same fixed points. From \eqref{eq:lambda_update},
\begin{equation}
g(r)
=
p_h r+(1-p_h)\left[\left(1-\frac{r}{\alpha}\right)+\frac{r}{\alpha}(1-\gamma)r\right],
\end{equation}
so $g'$ is affine in $r$ and $f'(r)=(1-\rho)+\rho g'(r)$ is affine and
increasing. Its extreme values on $[0,1]$ are therefore attained at the
endpoints:
\begin{equation}
f'(1)=1-\rho(1-p_h)\Bigl[1-\frac{1-2\gamma}{\alpha}\Bigr]=1-\rho\kappa,
\qquad
f'(0)=1-\rho(1-p_h)\Bigl[1+\frac{1}{\alpha}\Bigr].
\end{equation}
Since $\alpha\ge1$ and $\gamma>0$ we have $\kappa>0$, and since $\rho\le\tfrac12$ we
have $f'(0)\ge1-2\rho\ge0$. Hence
$L:=\sup_{r\in[0,1]}|f'(r)|=f'(1)=1-\rho\kappa\in(0,1)$, and by the mean value
theorem
\begin{equation}
|f(r)-r^*_{lin}(\alpha)|=|f(r)-f(r^*_{lin}(\alpha))|\le L\,|r-r^*_{lin}(\alpha)| .
\label{eq:linear_contraction}
\end{equation}

Define the innovation $\eta_{t+1}:=r_{t+1}-f(r_t)=\rho(\lambda_t-g(r_t))$, so
that $\E[\eta_{t+1}\mid r_t]=0$. Since $\lambda_t\in[0,1]$, its conditional
variance is at most $1/4$ by Popoviciu's inequality, and therefore
\begin{equation}
\E[\eta_{t+1}^2\mid r_t]\le\frac{\rho^2}{4}.
\label{eq:linear_noise_bound}
\end{equation}
Writing $r_{t+1}-r^*_{lin}(\alpha)=(f(r_t)-r^*_{lin}(\alpha))+\eta_{t+1}$ and using
orthogonality of the innovation,
\begin{equation}
\E\bigl[(r_{t+1}-r^*_{lin}(\alpha))^2\mid r_t\bigr]
\le
L^2 (r_t-r^*_{lin}(\alpha))^2+\frac{\rho^2}{4}.
\end{equation}
Taking expectations and iterating gives
$\E[(r_t-r^*_{lin}(\alpha))^2]\le L^{2t}(r_0-r^*_{lin}(\alpha))^2+\rho^2/\bigl(4(1-L^2)\bigr)$,
and since $1-L^2=\rho\kappa(2-\rho\kappa)$,
\begin{equation}
\limsup_{t\to\infty}\E\bigl[(r_t-r^*_{lin}(\alpha))^2\bigr]
\le
\frac{\rho^2}{4(1-L^2)}
=
\frac{\rho}{4\kappa(2-\rho\kappa)}.
\label{eq:linear_ms_bound}
\end{equation}
Finally \eqref{eq:linear_prob_concentration} follows from Markov's inequality
applied to $(r_t-r^*_{lin}(\alpha))^2$.
\end{proof}

The bound \eqref{eq:linear_ms_bound} is the standard rate for a
constant-step-size stochastic approximation: the process is pulled toward
$r^*_{lin}(\alpha)$ at rate at least $\rho\kappa$ and perturbed by innovations of size $\rho$,
so the two balance at a dispersion of order $\sqrt\rho$. It is informative only
for $\varepsilon\gg\sqrt\rho$.  The bound also controls the payoff. Under \eqref{eq:linear_policy} the lying
rate at score $r$ is $r/\alpha$, so by \eqref{eq:drift_reward} a period that
starts at score $r$ has expected sales
$\pi(r,r/\alpha)=\bigl(p_h+(1-p_h)r/\alpha\bigr)q(r)$, which is increasing and
convex in $r$, hence Lipschitz with constant
$L_\alpha:=(1-p_h)\bigl[\tfrac1\alpha+\tfrac1M+\tfrac{1-p_h}{\alpha p_hM}\bigr]$,
its slope at $r=1$. At the fixed point the drift vanishes,
$b\bigl(r^*_{lin}(\alpha),r^*_{lin}(\alpha)/\alpha\bigr)=0$, and $\hat u(r)$ is
the only lying rate with zero drift at $r$ by \eqref{eq:u_hat}, so
$r^*_{lin}(\alpha)/\alpha=\hat u\bigl(r^*_{lin}(\alpha)\bigr)$ and
$\pi\bigl(r^*_{lin}(\alpha),r^*_{lin}(\alpha)/\alpha\bigr)=\bar\pi\bigl(r^*_{lin}(\alpha)\bigr)$.
Since $\E[a_t\mid r_t]=\pi(r_t,r_t/\alpha)$, the Lipschitz property, Jensen's
inequality and \eqref{eq:linear_ms_bound} give, for every $r_0$,
\begin{equation}
\limsup_{T\to\infty}
\Bigl|\frac1T\sum_{t<T}\E[a_t]-\bar\pi\bigl(r^*_{lin}(\alpha)\bigr)\Bigr|
\le L_\alpha\sqrt{\frac{\rho}{4\kappa(2-\rho\kappa)}}=O(\sqrt\rho),
\label{eq:payoff_error}
\end{equation}
which justifies the steady-state approximation used in
Theorem~\ref{thm:optimal_linear_alpha} and the results that follow. The bound is
conservative, since deviations above and below $r^*_{lin}(\alpha)$ largely cancel;
in the simulations of Section~\ref{sec:num_classes} the error is of order $\rho$.

\begin{theorem}[Optimal linear-policy parameter]
\label{thm:optimal_linear_alpha}
Suppose the sender follows the linear policy \eqref{eq:linear_policy}, and the
receiver follows the probabilistic purchase rule
\eqref{eq:probabilistic_purchase}. Using the steady-state approximation \eqref{eq:payoff_error}, the sender's expected long-run
payoff is maximized at the operating level $r_{\bar\pi}$, the maximizer of $\bar\pi$ given in \eqref{eq:r_pibar}. This holds provided that $r_{\bar\pi}\in(0,1)$ and the corresponding parameter satisfies $\alpha^*\ge1$, where \begin{equation}
\alpha^*
=
\frac{
r_{\bar\pi}\bigl[1-(1-\gamma)r_{\bar\pi}\bigr]
}{
1-r_{\bar\pi}
}.
\label{eq:alpha_star_general}
\end{equation}
\end{theorem}

\begin{proof}

By \eqref{eq:payoff_error}, the sender's long-run payoff under
\eqref{eq:linear_policy} is approximated by $\bar\pi\bigl(r^*_{lin}(\alpha)\bigr)$,
so it suffices to write $\bar\pi$ at the fixed point as a function of $r$. The unconditional probability of a high signal is
\begin{equation}
\Prob(H)
=
p_h+(1-p_h)\frac{r}{\alpha}.
\label{eq:PH_linear}
\end{equation}
The receiver's purchase probability after observing $H$ is $q(r)$ of
\eqref{eq:probabilistic_purchase}.
 Inverting the fixed-point equation \eqref{eq:linear_fp_quadratic} for
$\alpha$ gives
\begin{equation}
\alpha(r)
=
\frac{
r\bigl[1-(1-\gamma)r\bigr]
}{
1-r
},
\qquad\text{so that}\qquad
\frac{r}{\alpha(r)}=\frac{1-r}{1-(1-\gamma)r}.
\label{eq:alpha_of_r}
\end{equation}

Substituting \eqref{eq:alpha_of_r} into \eqref{eq:PH_linear} gives
\begin{equation}
\Prob(H)
=
p_h+(1-p_h)
\frac{1-r}{1-(1-\gamma)r}=
\frac{
1+r(p_h\gamma-1)
}{
1-(1-\gamma)r
}.\label{PH_linear_reduced}
\end{equation}

Combining \eqref{PH_linear_reduced} and \eqref{eq:probabilistic_purchase}, the sender's
steady-state objective can be written as a function of $r$:
\begin{equation}
\bar\pi(r)
=
\frac{
p_h(1+c_o)
\bigl[1+r(p_h\gamma-1)\bigr]
}{
\bigl[1-(1-\gamma)r\bigr]
\left[
p_h(1+c_o)
+
(1-p_h)(1-r)c_l
\right]
}.
\label{eq:J_linear_r}
\end{equation}

By differentiating \eqref{eq:J_linear_r} with respect to $r$ and setting the
derivative equal to zero, we obtain the first-order condition
\begin{equation}
(1-\gamma)(1-p_h\gamma)r^2
-2(1-\gamma)r
+
\left[1-\gamma+p_h\gamma-M p_h\gamma\right]=F(r,M)
=0.
\label{eq:linear_foc}
\end{equation}
The smaller root of \eqref{eq:linear_foc} is \eqref{eq:r_pibar}.
The radicand is nonnegative precisely when
$(1-p_h\gamma)(1+c_o)\ge\gamma(1-p_h)c_l$, and the larger root exceeds
$1/(1-p_h\gamma)>1$; hence \eqref{eq:r_pibar} is the only critical
point in $(0,1)$. The derivative of \eqref{eq:J_linear_r} equals a positive
factor times the left-hand side of \eqref{eq:linear_foc}, which is equal to
$F(r,M)$ given in \eqref{F_long_form}, and $\partial F/\partial r=2(1-\gamma)\bigl[(1-p_h\gamma)r-1\bigr]<0$
on $[0,1]$. Hence, if $r_{\bar\pi}\in(0,1)$, the objective increases on
$(0,r_{\bar\pi})$ and decreases on $(r_{\bar\pi},1)$, so $r_{\bar\pi}$ maximizes it
over $(0,1)$. Since every $\alpha\ge1$ operates at some $r^*_{lin}(\alpha)\in(0,1)$
by Theorem~\ref{thm:linear_fixed_point_concentration}, the parameter that
produces $r_{\bar\pi}$ is optimal whenever it satisfies $\alpha\ge1$. Finally,
substituting $r_{\bar\pi}$ into \eqref{eq:alpha_of_r} gives
\eqref{eq:alpha_star_general}.
\end{proof}

Theorem~\ref{thm:optimal_linear_alpha} presupposes that $r_{\bar\pi}$ lies in
   $(0,1)$ and that the parameter $\alpha(r_{\bar\pi})$ is at least one. The next
   proposition says when these conditions hold and what the best linear policy
   is when they fail.

\begin{proposition}[Reach of the linear class]
\label{prop:linear_regimes}
Under the assumptions of Theorem~\ref{thm:optimal_linear_alpha}:
\begin{enumerate}
\item[(i)] The map $\alpha(\cdot)$ of \eqref{eq:alpha_of_r} is strictly
increasing, because
\begin{equation}
\alpha'(r)=\frac{\gamma+(1-\gamma)(1-r)^2}{(1-r)^2}>0 ,
\label{eq:alpha_prime}
\end{equation}
so the attainable steady-state trust levels form the interval $[r^*_{lin}(1),1)$,
whose left endpoint is obtained at $\alpha=1$:
\begin{equation}
r^*_{lin}(1)=\frac{1}{1+\sqrt \gamma},
\qquad\text{and}\qquad
1-r^*_{lin}(\alpha)=\frac{\gamma}{\alpha}+O(\alpha^{-2})
\ \text{ as }\alpha\to\infty .
\label{eq:linear_endpoints}
\end{equation}
\item[(ii)] The critical point $r_{\bar\pi}$ is real and lies in $(0,1)$ if
and only if the tolerance satisfies
\begin{equation}
\gamma \;<\;M\;<\;\frac{1-\gamma (1-p_h)}{\gamma \,p_h}:=\overline{M}.
\label{eq:interior_condition}
\end{equation}
In that case, if $r_{\bar\pi}\ge r^*_{lin}(1)$ the optimal parameter is
\eqref{eq:alpha_star_general}, and if $r_{\bar\pi}<r^*_{lin}(1)$ the steady-state
objective is maximized at the most aggressive feasible linear policy,
$\alpha^*=1$, with operating level $r^*_{lin}(1)$.
\item[(iii)] If $M\le\gamma$, the steady-state objective is increasing
over $[r^*_{lin}(1),1)$, and its supremum is approached as $\alpha\to\infty$, with
$r^*_{lin}(\alpha)\to1$.
\item[(iv)] If $M\ge \overline{M}$, the
steady-state objective is decreasing over $[r^*_{lin}(1),1)$, and $\alpha^*=1$.
\end{enumerate}
\end{proposition}

\begin{proof}
(i) Writing $\alpha(r)=\bigl(r-(1-\gamma)r^2\bigr)/(1-r)$, the quotient
rule gives the numerator
$\bigl(1-2(1-\gamma)r\bigr)(1-r)+r-(1-\gamma)r^2=1-2(1-\gamma)r+(1-\gamma)r^2
=\gamma+(1-\gamma)(1-r)^2$, which is \eqref{eq:alpha_prime}. At $\alpha=1$ the
discriminant in \eqref{eq:linear_r_star} equals $4\gamma$, so
$r^*_{lin}(1)=(1-\sqrt\gamma)/(1-\gamma)=(1+\sqrt\gamma)^{-1}$; and
$\alpha(r)\to\infty$ as $r\to1$. Hence $\alpha(\cdot)$ maps $[r^*_{lin}(1),1)$ onto
$[1,\infty)$. For the asymptotics, \eqref{eq:alpha_of_r} gives
$1-r=r\bigl[1-(1-\gamma)r\bigr]/\alpha$, and
$r\bigl[1-(1-\gamma)r\bigr]=\gamma+O(1-r)$ as $r\to1$, which yields
\eqref{eq:linear_endpoints}.

For (ii)--(iv), by the proof of Theorem~\ref{thm:optimal_linear_alpha},
the derivative of the  objective has the sign of $F$ and $F$ is strictly decreasing on $[0,1]$,
with
\begin{equation}
    F(0,M)=1-\gamma(1-p_h)-p_h\gamma M,
\qquad
F(1,M)=p_h\gamma(\gamma-M).
\end{equation}
(ii) $F$ has a root in $(0,1)$, necessarily its smaller root $r_{\bar\pi}$,
if and only if $F(0,M)>0>F(1,M)$, which is \eqref{eq:interior_condition}. If
$r_{\bar\pi}\ge r^*_{lin}(1)$, then $\alpha^*=\alpha(r_{\bar\pi})\ge1$ by (i), and
Theorem~\ref{thm:optimal_linear_alpha} applies. If $r_{\bar\pi}<r^*_{lin}(1)$, the
objective is decreasing on $[r^*_{lin}(1),1)$, so the best attainable level is
$r^*_{lin}(1)$, with $\alpha^*=1$. (iii) If $M\le\gamma$, then $F(1,M)\ge0$, so
$F>0$ on $[0,1)$; the objective is increasing, and its supremum over
$[r^*_{lin}(1),1)$ is approached as $r\to1$, that is, as $\alpha\to\infty$ by (i).
(iv) If $M$ is at least the upper bound in \eqref{eq:interior_condition}, then
$F(0,M)\le0$, so $F<0$ on $(0,1]$; the objective is decreasing, and its
maximum over $[r^*_{lin}(1),1)$ is at $r^*_{lin}(1)$, with $\alpha^*=1$.
\end{proof}

The lower endpoint $r^*_{lin}(1)=(1+\sqrt \gamma)^{-1}$ has an interpretation from Section~\ref{sec:premium}: it equals $r^*(u^\sharp)=u^\sharp$, the resting score at the rate where the credibility premium is largest. The most aggressive feasible linear policy is the one that operates
where the gap between recorded and actual honesty is widest.

Next, we study how loss parameters affect the interior
solution, and the optimal policy parameter.

\begin{proposition}[Effect of receiver losses]
\label{prop:comparative_statics}
Suppose the interior solution $r_{\bar\pi}\in(0,1)$ is feasible. Then
$\partial r_{\bar\pi}/\partial M<0$ and $\partial\alpha^*/\partial M<0$.
Consequently,
\begin{equation}
    \frac{\partial r_{\bar\pi}}{\partial c_l}>0,
\quad
\frac{\partial r_{\bar\pi}}{\partial c_o}<0, \text{ and } \frac{\partial \alpha^*}{\partial c_l}>0,
\quad
\frac{\partial \alpha^*}{\partial c_o}<0.
\end{equation}
\end{proposition}

\begin{proof}
Recall $M=(1+c_o)/c_l$ from \eqref{eq:M_def}. By
\eqref{eq:linear_foc}, the interior optimizer $r_{\bar\pi}$ solves $ F(r,M)=0.$ By implicit differentiation,
\begin{equation}
    \frac{\partial r_{\bar\pi}}{\partial M}
=
-\frac{F_M}{F_r}
=
\frac{p_h\gamma}
{2(1-\gamma)\bigl((1-p_h\gamma)r_{\bar\pi}-1\bigr)}.
\end{equation}
Since $r_{\bar\pi}\in(0,1)$ and $1-p_h\gamma<1$, the denominator is negative, which
gives $\partial r_{\bar\pi}/\partial M<0$. Because $M$ is increasing in
$c_o$ and decreasing in $c_l$, we obtain the first two relations. To relate $r_{\bar\pi}$ to $\alpha^*$, recall from \eqref{eq:alpha_prime}
that $\alpha'(r)>0$. Therefore, $\alpha^*=\alpha(r_{\bar\pi})$ moves in the same
direction as $r_{\bar\pi}$. This gives the last two relations.
\end{proof}

These comparative statics have a natural interpretation. When $c_l$
increases, receivers become more reluctant to purchase after a high signal
because buying a low-quality product is more costly. Then, the sender must 
preserve a higher trust to remain persuasive. This increases both
the optimal steady-state trust level $r_{\bar\pi}$ and the optimal policy parameter
$\alpha^*$. 

In contrast, when $c_o$ increases, receivers become more concerned
about missing high-quality products. Therefore, they are willing to purchase
after favorable signals even at lower trust levels. The sender can exploit
trust more aggressively, which lowers the optimal trust level $r_{\bar\pi}$
and reduces $\alpha^*$. Both statements are the same comparative static:
what matters is the tolerance $M$, and the two loss parameters move it in
opposite directions.

\begin{proposition}[Symmetric-loss benchmark]
\label{prop:benchmark_linear}
Suppose $c_l=1$ and $c_o=0$, i.e. $M=1$. Then the
interior maximizer in Theorem~\ref{thm:optimal_linear_alpha} simplifies to
\begin{equation}
r_{\bar\pi}
=
\frac{1-\sqrt{p_h \gamma}}{1-p_h \gamma}
=
\frac{1}{1+\sqrt{p_h \gamma}},\label{simple_r}
\end{equation}
and the corresponding optimal linear-policy parameter is
\begin{equation}
\alpha^*
=
\frac{\sqrt{p_h}+\sqrt{\gamma}}
{\sqrt{p_h}\left(1+\sqrt{p_h \gamma}\right)}. \label{simple_alpha}
\end{equation}
\end{proposition}

\begin{proof}
Setting $c_l=1$ and $c_o=0$ in \eqref{eq:r_pibar}
yields \eqref{simple_r}. Substituting this into
\eqref{eq:alpha_star_general} and simplifying yields \eqref{simple_alpha}.
Since $r_{\bar\pi}=(1+\sqrt{p_h\gamma})^{-1}\ge(1+\sqrt \gamma)^{-1}=r^*_{lin}(1)$, the solution is always feasible.
\end{proof}

Proposition~\ref{prop:benchmark_linear} provides a useful benchmark. The
optimal steady-state trust level depends only on $\sqrt{p_h\gamma}$, and the
lying rate that sustains it is
$\hat u(r_{\bar\pi})=\sqrt{p_h}/(\sqrt{p_h}+\sqrt\gamma)$. As $p_h$ increases,
the sender maintains a lower trust level and misrepresents more often at it.
As $\gamma$ increases, the sender also maintains a lower trust level, but
misrepresents less often. The reason is visible in the two factors of
\eqref{eq:common_value}: raising $\gamma$ lowers the high-signal factor
$\bigl(1-(1-p_h\gamma)r\bigr)/\bigl(1-(1-\gamma)r\bigr)$ at every interior
$r$, so holding a given score requires giving up more high signals, while the
purchase factor $q$ is unaffected by $\gamma$. Better detection thus makes
credibility more expensive to hold, and the sender settles at a lower level;
because each lie is now recorded more often, it lies less even there.

\begin{proposition}[Extreme receiver preferences]
\label{prop:loss_limits}
Consider the optimal linear policy characterized in
Theorem~\ref{thm:optimal_linear_alpha}.
\begin{enumerate}
\item As the lemon loss vanishes (\(c_l\to 0\)), we have
\((r^*_{lin}(\alpha^*),\alpha^*)\to(\frac{1}{1+\sqrt{\gamma}},1)\); conversely, as
\(c_l\to\infty\), we have \((r^*_{lin}(\alpha^*),\alpha^*)\to(1,\infty)\).
\item As the opportunity loss \(c_o\to\infty\),
\((r^*_{lin}(\alpha^*),\alpha^*)\to(\frac{1}{1+\sqrt{\gamma}},1)\).
\end{enumerate}
Both regimes are governed by the tolerance $M$: the limits
$c_l\to0$ and $c_o\to\infty$ are $M\to\infty$, and the limit $c_l\to\infty$ is
$M\to0$.
\end{proposition}

\begin{proof}
The limits $c_l\to0$ and $c_o\to\infty$ send $M\to\infty$, so eventually $M$
exceeds the upper bound in \eqref{eq:interior_condition}, and
Proposition~\ref{prop:linear_regimes}(iv) gives $\alpha^*=1$ with operating
level $r^*_{lin}(1)=(1+\sqrt\gamma)^{-1}$. The limit $c_l\to\infty$ sends $M\to0$,
so eventually $M\le\gamma$, and Proposition~\ref{prop:linear_regimes}(iii)
shows that the supremum is approached as $\alpha\to\infty$, with
$r^*_{lin}(\alpha)\to1$.
\end{proof}

Proposition~\ref{prop:loss_limits} illustrates how extreme receiver loss shapes the sender's optimal behavior, and corresponding long-term trust.
When the lemon loss is low, the role of trust
diminishes. The sender follows an aggressive policy, which drives trust to
its lowest sustainable level. In contrast, when buying a lemon is
very costly, the sender's optimal honesty rate rises to its
maximum. An opposite effect occurs for opportunity loss.
As it increases, the receiver is willing to buy any product which
the sender reports as high quality, for fear of losing high
quality ones. The sender exploits this, and deceives at its maximum capacity.

\subsection{$(s,S)$ Policy}

Next, we study an $(s,S)$ type policy, defined in \eqref{eq:sS_policy}. We show that, under a regenerative-cycle
approximation, this build-and-exploit cycle does not outperform a threshold
policy. For this policy type, the analysis needs to proceed differently from the linear policy. Under
$(s,S)$ and threshold policies the sender's misrepresentation rate is either $0$
or $1$, so by Lemma~\ref{lem:rest_point} neither
mode has an interior rest point. Instead of settling at a fixed level, trust
travels upward while the sender is honest and downward while it is exploiting.
Therefore, the long-run payoff can  be computed over a cycle rather than at a
fixed point, which motivates the regenerative approximation below.

\subsubsection{Regenerative-Cycle Approximation for the $(s,S)$ Policy}

Under the $(s,S)$ policy the score moves in two phases. In the
\emph{building phase} the sender is honest, $u=0$, and by
\eqref{eq:drift_reward} the score drifts upward at rate
$\rho\,b(r,0)=\rho(1-p_h)(1-r)$, earning $\pi(r,0)=p_hq(r)$ per period. In the
\emph{exploitation phase} the sender lies, $u=1$, the score drifts downward at
rate $\rho\,|b(r,1)|=\rho(1-p_h)\gamma r$, and the sender earns $\pi(r,1)=q(r)$
per period. The \emph{regenerative-cycle approximation} replaces the random
path within each phase by its mean drift, so that one cycle carries the score
from $s$ up to $S$ and back down to $s$, and evaluates the policy by the expected
reward of one cycle divided by its expected length. For $0<s<S<1$ this gives
\begin{equation}
\hat{\mathcal J}(s,S)
:=
\frac{\displaystyle\int_s^S\Bigl[\frac{\pi(r,0)}{b(r,0)}+\frac{\pi(r,1)}{|b(r,1)|}\Bigr]dr}
{\displaystyle\int_s^S\Bigl[\frac{1}{b(r,0)}+\frac{1}{|b(r,1)|}\Bigr]dr}
=
\frac{\displaystyle\int_s^S q(r)\Bigl[\frac{p_h}{1-r}+\frac{1}{\gamma r}\Bigr]dr}
{\displaystyle\int_s^S\Bigl[\frac{1}{1-r}+\frac{1}{\gamma r}\Bigr]dr},
\label{eq:regenerative_approx}
\end{equation}
since the time spent in $[r,r+dr]$ is $dr/(\rho\,b(r,0))$ on the way up and
$dr/(\rho\,|b(r,1)|)$ on the way down, and $\rho$ cancels. A degenerate band
$s=S=\tau$ is evaluated by the $S\downarrow s$ limit of
\eqref{eq:regenerative_approx}, which the proof below computes directly; for
$s=0$ or $S=1$ one phase never ends and the policy is the constant rate
$u\equiv1$ or $u\equiv0$.

\begin{theorem}[Collapse of the \texorpdfstring{$(s,S)$}{s,S} band]
\label{thm:sS_collapse_threshold}
Under the regenerative-cycle approximation \eqref{eq:regenerative_approx}, every nondegenerate $(s,S)$ policy is
weakly dominated by a degenerate band. Hence, there exists an optimal
$(s,S)$ policy satisfying
\begin{equation}
s^*=S^*=\tau^* ,
\end{equation}
where $\tau^* \in [0,1] $ is the optimal cutoff of the threshold policy \eqref{eq:threshold_policy}. Consequently, the optimal $(s,S)$ policy is
equivalent to a threshold policy in terms of the sender's outcome.
\end{theorem}

\begin{proof}
Define
\begin{equation}
W(r)
:=
\frac{1}{1-r}
+
\frac{1}{\gamma r}, \quad 
G(r)
:=
q(r)
\frac{
\frac{p_h}{1-r}
+
\frac{1}{\gamma r}
}{
\frac{1}{1-r}
+
\frac{1}{\gamma  r}
}
\;=\;q(r)\,\frac{1-r+p_h\gamma r}{1-r+\gamma r} .\label{G_r}
\end{equation}

This is the stationary value $\bar\pi(r)$ of \eqref{eq:common_value}, since
$1-(1-p_h\gamma)r=1-r+p_h\gamma r$ and $D(r)=1-r+\gamma r$. Since $\gamma \in(0,1]$ and $r\in(0,1)$, we have $W(r)>0$. Moreover, by the definition
of $G(r)$ in \eqref{G_r},
\begin{equation}
q(r)
\left[
\frac{p_h}{1-r}
+
\frac{1}{\gamma r}
\right]
=
W(r)G(r).
\end{equation}
Thus, by \eqref{eq:regenerative_approx}, we write the sender's approximate payoff as
\begin{equation}
\hat{\mathcal J}(s,S)
=
\frac{
\int_s^S W(r)G(r)\,dr
}{
\int_s^S W(r)\,dr
}.
\end{equation}
Because $W(r)>0$, this expression is a weighted average of $G(r)$ over the
interval $[s,S]$. Therefore,
\begin{equation}
\min_{r\in[s,S]}G(r)
\le
\hat{\mathcal J}(s,S)
\le
\max_{r\in[s,S]}G(r).
\end{equation}

It remains to evaluate the payoff of a degenerate band $s=S=\tau$ directly,
since in that case trust does not traverse an interval but alternates around
$\tau$. Near $\tau$, upward moves occur at rate $(1-p_h)$ with size
$\rho(1-\tau)$, and downward moves occur at rate
$(1-p_h)\gamma $ with size $\rho\tau$, by \eqref{eq:lambda_update}. The long-run fraction of exploitation periods is the long-run lying rate $u$,
and stationarity of $r_t$ requires the two flows to balance,
\begin{equation}
(1-u)(1-\tau)=u\,\gamma\,\tau
\qquad\Longrightarrow\qquad
u=\frac{1-\tau}{1-\tau+\gamma\tau}=\hat u(\tau),
\label{eq:flow_balance}
\end{equation}
the rate of \eqref{eq:u_hat} that holds the score at $\tau$. The sender earns
$p_hq(\tau)$ in building periods and $q(\tau)$ in exploitation periods, so the
degenerate band earns
\begin{equation}
\bigl(1-\hat u(\tau)\bigr)p_h q(\tau)+\hat u(\tau)\,q(\tau)
=
q(\tau)\,\frac{1-\tau+p_h\gamma \tau}{1-\tau+\gamma \tau}
=
G(\tau),
\label{eq:degenerate_payoff}
\end{equation}
which agrees with the $S\downarrow s$ limit of the weighted average, since
$G$ is continuous on $(0,1)$.  Let $\tau\in\arg\max_{r\in[s,S]}G(r)$. By
\eqref{eq:degenerate_payoff},
\begin{equation}
\hat{\mathcal J}(s,S)\le G(\tau).
\end{equation}
Thus, every nondegenerate band is weakly dominated by a degenerate band
inside the same interval. It follows that the maximization problem over $0\le s\le S\le 1$ admits an optimizer on
$s=S$. Hence, there exists an optimal solution which satisfies
\begin{equation}
s^*=S^*=\tau^*, \text{ where }  \tau^*
\in
\arg\max_{\tau\in[0,1]}G(\tau).
\end{equation} \end{proof}

Theorem \ref{thm:sS_collapse_threshold} shows that, under the regenerative-cycle approximation, the
$(s,S)$ policy does not generate a strict benefit from cycling through a range of trust values. The reasoning comes from the cycle payoff approximation. The cycle payoff is a
weighted average of local payoffs along the interval $[s,S]$. A weighted average
cannot exceed the point with the maximum payoff inside the interval. Thus,
instead of traveling inside the interval, the sender can operate near the trust
level that gives the largest expected payoff, which is weakly better than
moving through many trust levels.

This finding differs from inventory control problems, where $(s,S)$
policies are generally valuable because of a fixed ordering cost. In that
scenario, it is inefficient to adjust the inventory level continuously, so the
decision-maker is better off waiting until inventory is sufficiently low.

In our model, an analogous fixed cost might represent a marketing
cost that is spent to change the `information conveying strategy' from a more
frequent honesty to deception, or the opposite. However, because there is no
fixed cost or adjustment delay, a wide band is not valuable in this scenario. The absence of a switching cost, rather than the speed of the trust update,
is what removes the benefit of a band. As a result, the optimal band collapses to a
threshold.

The converse confirms this reading. Suppose the sender pays a fixed cost
$\eta>0$ each time it switches mode. One cycle then contains two switches, and
the cycle payoff becomes $\bigl(\int_s^S W G\,dr-2\rho(1-p_h)\eta\bigr)/\int_s^S W\,dr$. Since a degenerate band
switches modes in a positive fraction of periods while a wide band switches only
twice per cycle, we expect nondegenerate bands to become strictly optimal
once $\eta$ is large enough; we do not prove this, and state it as a conjecture
supported by the switching-cost experiment in Section~\ref{sec:numerical}. The
collapse in Theorem~\ref{thm:sS_collapse_threshold} is therefore a consequence of
costless mode switching in our environment. Nor is it an artifact of the regenerative approximation: Theorem~\ref{thm:fluid} bounds the payoff of every policy---band-based or otherwise---by $\max_r\bar\pi(r)$ without using the cycle decomposition at all, and Corollary~\ref{cor:finite_rho} does the same at a fixed $\rho$ up to an $O(\rho)$ term.

\subsection{Threshold Policy}

When the band collapses to a single cutoff $s=S=\tau$, the cycle of
\eqref{eq:regenerative_approx} degenerates and the score alternates around
$\tau$. The following theorem characterizes the maximizer of the value \eqref{eq:degenerate_payoff}.

\begin{theorem}[Optimal threshold]
\label{thm:optimal_threshold}
The threshold-policy value function \eqref{eq:degenerate_payoff} has a unique interior
critical point in $(0,1)$, which is its maximizer, if and only if the tolerance satisfies \eqref{eq:interior_condition}, given as
\begin{equation}
\tau^*
=
\frac{
c_l(1-\gamma )
-
\sqrt{
\gamma \,p_h\,c_l(1-\gamma )
\big[(1-p_h \gamma )(1+c_o)-\gamma (1-p_h)c_l\big]
}
}{
c_l(1-\gamma )(1-p_h \gamma )
}.
\label{eq:tau_star}
\end{equation}
If $M\le \gamma$ the maximizer is $\tau^*=1$, and if $M\geq \overline{M},$ then the maximizer is $\tau^*=0$.
\end{theorem}

\begin{proof}
By \eqref{G_r}, the value \eqref{eq:degenerate_payoff} of the threshold policy
is $G(\tau)=\bar\pi(\tau)$, which, written out with \eqref{eq:probabilistic_purchase},
is the function \eqref{eq:J_linear_r} in the proof of
Theorem~\ref{thm:optimal_linear_alpha} with $r$ replaced by $\tau$. That proof
shows that $G'(\tau)$ equals a positive factor times $F(\tau,M)$, the quadratic
in \eqref{eq:linear_foc}, and that
$\partial F/\partial\tau=2(1-\gamma)\bigl[(1-p_h\gamma)\tau-1\bigr]<0$ on
$[0,1]$, so $F(\cdot,M)$ is strictly decreasing there. The sign of $G'$ on
$[0,1]$ is therefore determined by the endpoint values computed in the proof of Proposition~\ref{prop:linear_regimes}.

If $M\le\gamma$, then $F(1,M)\ge0$, so $F>0$ on $[0,1)$ and $G$ is strictly
increasing on $[0,1]$: the maximizer is $\tau^*=1$ and there is no interior
critical point. If $M\ge\overline M$, then $F(0,M)\le0$, so $F<0$ on $(0,1]$
and $G$ is strictly decreasing: the maximizer is $\tau^*=0$ and again there is
no interior critical point. Otherwise $F(0,M)>0>F(1,M)$, which is exactly
\eqref{eq:interior_condition}, and $F(\cdot,M)$ has a unique root in $(0,1)$,
with $G$ increasing to the left of it and decreasing to the right; this root is
the unique interior critical point of $G$ and its maximizer. By the proof of
Theorem~\ref{thm:optimal_linear_alpha}, the larger root of
\eqref{eq:linear_foc} exceeds $1/(1-p_h\gamma)>1$, so the root in $(0,1)$ is
the smaller one, \eqref{eq:r_pibar}.

It remains to express \eqref{eq:r_pibar} in terms of the losses. Substituting
$M=(1+c_o)/c_l$ into the radicand of \eqref{eq:r_pibar} gives
\begin{equation}
\frac{p_h\gamma\bigl[(1-p_h\gamma)M-\gamma(1-p_h)\bigr]}{1-\gamma}
=
\frac{p_h\gamma\bigl[(1-p_h\gamma)(1+c_o)-\gamma(1-p_h)c_l\bigr]}{c_l(1-\gamma)} .
\end{equation}
Multiplying the numerator and the denominator of \eqref{eq:r_pibar} by
$c_l(1-\gamma)$ and moving the factor $c_l(1-\gamma)$ under the square root
yields \eqref{eq:tau_star}.
\end{proof}

\begin{remark}[When credibility is worth managing]
\label{rem:threshold_corners}
Condition \eqref{eq:interior_condition} is stated entirely in the tolerance $M$
and delimits the regime in which credibility is worth managing. If $M\le \gamma $, the
receivers are so cautious that the sender never gains from misrepresentation and
maintains maximal credibility, $\tau^*=1$. If
$M\ge \overline{M}$, the receivers are so willing to buy that the
sender misrepresents in every low state, $\tau^*=0$. Thus a large $c_l$ pushes $\tau^*$ toward $1$ and a large $c_o$ pushes it toward $0$, mirroring the comparative statics of the
linear policy. The benchmark $M=1$ of
Proposition~\ref{prop:benchmark_linear} always satisfies
\eqref{eq:interior_condition}, since $\gamma <1$.
\end{remark}

Two degenerate policies provide natural benchmarks. Evaluating
\eqref{eq:degenerate_payoff} at the endpoints gives
\begin{equation}
\bar\pi(1)=p_h
\qquad\text{(always honest)},
\qquad
\bar\pi(0)=q(0)=\frac{p_h(1+c_o)}{p_h(1+c_o)+(1-p_h)c_l}
\qquad\text{(always deceptive)} ,
\label{eq:degenerate_benchmarks}
\end{equation}
and by Proposition~\ref{prop:receiver_models}(ii) $\bar\pi(0)>\bar\pi(1)$ if and only if $M>1$. The
value of credibility management,
$\bar\pi(\tau^*)-\max\{\bar\pi(0),\bar\pi(1)\}$, is therefore strictly positive
precisely under \eqref{eq:interior_condition} and vanishes in the corner regimes
of Remark~\ref{rem:threshold_corners}.

\begin{remark}[Attaining the fluid bound]
\label{rem:threshold_attains}
In the fluid model \eqref{eq:fluid_model}, the threshold and the linear policy drive the score to a
single level at which the lying rate is $\hat u$, and they therefore earn $\bar\pi$
there. Under the threshold policy with $\tau\in(0,1)$, $b(r,0)>0$ below $\tau$ and
$b(r,1)<0$ at and above it, both bounded away from zero. Every trajectory therefore
reaches $\tau$ in finite time and stays there as the mixture of $u=0$ and $u=1$
with weight $\hat u(\tau)$ on $u=1$, a Filippov sliding mode that matches
\eqref{eq:flow_balance}. Since $\pi$ is affine in $u$, the reward there is
$\bar\pi(\tau)$. For $\tau\in\{0,1\}$ the policy is the constant rate $u\equiv1$ or
$u\equiv0$. Under the linear policy, $b(r,r/\alpha)$ is a positive multiple of the
quadratic in \eqref{eq:linear_fp_quadratic}, which is positive below
$r^*_{lin}(\alpha)$ and negative above it. Every trajectory therefore converges to
$r^*_{lin}(\alpha)$, where $r/\alpha=\hat u(r)$. In both cases the transient
contributes $O(1/T)$. At $\tau^*$ and $\alpha^*$ the limit level is $r_{\bar\pi}$
(Lemma~\ref{lem:pibar_shape}).
\end{remark}

\subsection{Equivalence and Dominance Across Policy Classes}
\label{sec:equivalence}

The three policy classes have so far been analyzed separately, but they solve
the same problem. The linear policy at its fixed point earns
$\bar\pi\bigl(r^*_{lin}(\alpha)\bigr)$, because its lying rate there is
$r^*_{lin}(\alpha)/\alpha=\hat u\bigl(r^*_{lin}(\alpha)\bigr)$. The threshold
policy at its cutoff earns $G(\tau)=\bar\pi(\tau)$ by
\eqref{eq:degenerate_payoff}, because it alternates between $u=0$ and $u=1$ in
the proportions \eqref{eq:flow_balance} that hold the score at $\tau$. An $(s,S)$
band earns a weighted average of $\bar\pi$ over $[s,S]$. All three therefore
reduce to the choice of the level at which the score is held, and
Lemma~\ref{lem:pibar_shape} identifies the best such level, $r_{\bar\pi}$. The
reason is that drift and reward are affine in the lying rate, by
\eqref{eq:drift_reward}. How a policy produces the average rate $\hat u(r)$ at
the level it holds---by an interior rate in every period, or by alternating
between honesty and exploitation---does not affect what it earns there.

What distinguishes the classes is which levels they can hold, and this has a
simple geometric form. For a stationary policy $u=\sigma(r)$, the drift factors as
\begin{equation}
b\bigl(r,\sigma(r)\bigr)=(1-p_h)\,D(r)\,\bigl[\hat u(r)-\sigma(r)\bigr],
\label{eq:drift_factored}
\end{equation}
so the score rises where the graph of $\sigma$ lies below the zero-drift curve
$u=\hat u(r)$ and falls where it lies above. A nondecreasing $\sigma$ crosses
this strictly decreasing curve exactly once, and the crossing is the level it
sustains. The best level is the \emph{target point}
$z^\star:=\bigl(r_{\bar\pi},\hat u(r_{\bar\pi})\bigr)$ on the curve
(Figure~\ref{fig:target_point}). The curve depends on $\gamma$ alone, and the
tolerance $M$ only moves $z^\star$ along it: up and to the left as $M$ increases, by
Proposition~\ref{prop:comparative_statics}.

The two classes reach the curve differently. A threshold rule jumps from $u=0$ to
$u=1$ at $\tau$, so it crosses the curve at $\tau$ for every $\tau\in[0,1]$. A
linear rule $u=r/\alpha$ is a ray from the origin with slope $1/\alpha\le1$, so it
crosses only at points on or below the diagonal $u=r$. The diagonal meets the
curve at $r=(1+\sqrt\gamma)^{-1}$. This is the left end $r^*_{lin}(1)$ of the
interval of levels the linear class can hold, and it is also the lying rate
$u^\sharp$ at which the credibility premium of Section~\ref{sec:premium}
is widest.

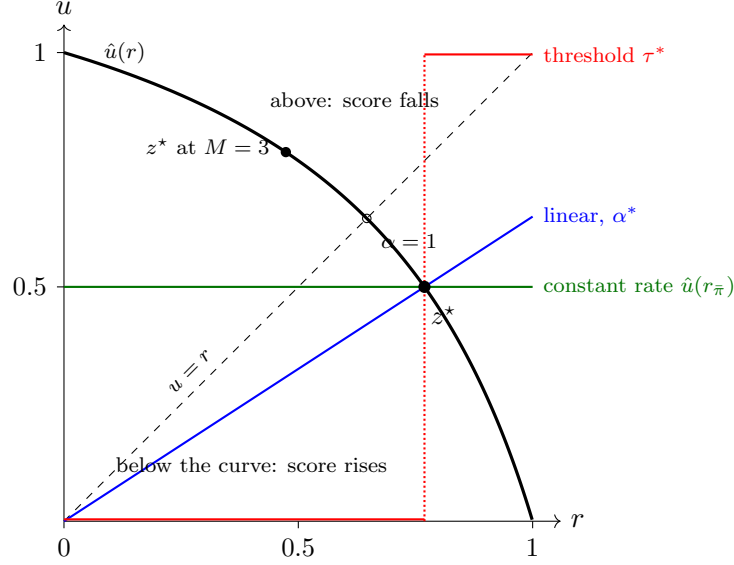
\begin{figure}[t]
\centering
\begin{tikzpicture}[scale=6.2]
\draw[->] (0,0) -- (1.06,0) node[right] {$r$};
\draw[->] (0,0) -- (0,1.06) node[above] {$u$};
\foreach \x in {0,0.5,1} \draw (\x,0) -- (\x,-0.015) node[below]{\small $\x$};
\foreach \y in {0.5,1} \draw (0,\y) -- (-0.015,\y) node[left]{\small $\y$};
\draw[very thick,domain=0:1,samples=200] plot (\x,{(1-\x)/(1-0.7*\x)});
\node at (0.13,1.0) {\scriptsize $\hat u(r)$};
\draw[dashed] (0,0) -- (1,1);
\node[rotate=45] at (0.27,0.32) {\scriptsize $u=r$};
\draw[blue,thick] (0,0) -- (1,0.65) node[right]{\scriptsize linear, $\alpha^*$};
\draw[green!45!black,thick] (0,0.5) -- (1,0.5) node[right]{\scriptsize constant rate $\hat u(r_{\bar\pi})$};
\draw[red,thick] (0,0.004) -- (0.7692,0.004);
\draw[red,thick] (0.7692,0.996) -- (1,0.996) node[right]{\scriptsize threshold $\tau^*$};
\draw[red,thick,densely dotted] (0.7692,0) -- (0.7692,1);
\fill (0.7692,0.5) circle (0.013);
\node at (0.81,0.44) {\small $z^\star$};
\fill (0.4734,0.7876) circle (0.011);
\node[left] at (0.46,0.80) {\scriptsize $z^\star$ at $M=3$};
\draw (0.6461,0.6461) circle (0.009);
\node[below right] at (0.655,0.635) {\scriptsize $\alpha=1$};
\node at (0.40,0.12) {\scriptsize below the curve: score rises};
\node at (0.62,0.90) {\scriptsize above: score falls};
\end{tikzpicture}
\caption{The zero-drift curve $u=\hat u(r)$ of \eqref{eq:u_hat} at
$p_h=\gamma=0.3$, $M=1$. Below the curve the score rises and above it the
score falls, so a nondecreasing policy settles where its graph crosses the
curve. The target point $z^\star=(r_{\bar\pi},\hat u(r_{\bar\pi}))=(0.769,0.5)$
is the best such level: its first coordinate is the score $r_{\bar\pi}$ that
maximizes $\bar\pi$, and its second is the lying rate that holds the score
there, $\hat u(r_{\bar\pi})=\sqrt{p_h}/(\sqrt{p_h}+\sqrt\gamma)$ by
Proposition~\ref{prop:benchmark_linear}, which equals $0.5$ here only because
$p_h=\gamma$. The optimal linear ray $u=r/\alpha^*$, the optimal threshold,
whose jump fills the vertical segment at $\tau^*$, and the constant rate
$u\equiv\hat u(r_{\bar\pi})$ all pass through $z^\star$: each holds the score
at $r_{\bar\pi}$, and by \eqref{eq:u_hat} any policy that does so must lie at
the average rate $\hat u(r_{\bar\pi})$. The open circle marks the linear policy
with $\alpha=1$, whose graph is the diagonal $u=r$; it meets the curve at
$r=(1+\sqrt\gamma)^{-1}=0.646$, the lowest level the linear class can hold. At
$M=3$ the target moves to $(0.473,0.788)$, above the diagonal and out of reach
of every linear policy with $\alpha\ge1$, so the best linear policy settles at
the open circle instead.}
\label{fig:target_point}
\end{figure}

\begin{proposition}[Payoff equivalence and strict dominance]
\label{thm:equivalence}
Under the steady-state approximation \eqref{eq:payoff_error} and the regenerative-cycle approximation \eqref{eq:regenerative_approx}:
\begin{enumerate}
\item[(i)] the linear policy at its fixed point and the threshold policy at its
cutoff both earn $\bar\pi$ at their operating level, and an $(s,S)$ band earns a
weighted average of $\bar\pi$ over $[s,S]$;
\item[(ii)] the threshold class attains every level in $[0,1]$, so
$\tau^*=r_{\bar\pi}$, whereas the linear class attains exactly
$\bigl[(1+\sqrt\gamma)^{-1},1\bigr)$;
\item[(iii)] the threshold policy weakly dominates the linear policy, and strictly
if and only if
\begin{equation}
M>M_{\max}:=\frac{1-\sqrt \gamma +2p_h\sqrt \gamma }{p_h\bigl(1+\sqrt \gamma \bigr)} .
\label{eq:M_max}
\end{equation}
For $M\le M_{\max}$ the two classes earn the same value. It is attained at
$\alpha^*$ when $\gamma<M\le M_{\max}$ and approached as $\alpha\to\infty$ when
$M\le\gamma$.
\end{enumerate}
Moreover,
\begin{equation}
M_{\max}-1=\frac{(1-\sqrt\gamma)(1-p_h)}{p_h\bigl(1+\sqrt\gamma\bigr)}>0
\qquad\text{for all }p_h,\gamma\in(0,1),
\label{eq:Mmax_gap}
\end{equation}
so the benchmark $M=1$ of Proposition~\ref{prop:benchmark_linear} always lies
strictly inside the region where the two classes are payoff-equivalent.
\end{proposition}

\begin{proof}
(i) At the fixed point, $r^*_{lin}(\alpha)/\alpha=\hat u(r^*_{lin}(\alpha))$, so
the steady-state payoff of the linear policy is
$\pi\bigl(r^*_{lin}(\alpha),\hat u(r^*_{lin}(\alpha))\bigr)
=\bar\pi\bigl(r^*_{lin}(\alpha)\bigr)$. The statements for the threshold and
$(s,S)$ classes are \eqref{eq:degenerate_payoff} and the weighted-average
representation in the proof of Theorem~\ref{thm:sS_collapse_threshold}, since
$G=\bar\pi$.

(ii) The attainable sets are given by Proposition~\ref{prop:linear_regimes}(i)
and Theorem~\ref{thm:optimal_threshold}.

(iii) Let $r_1:=(1+\sqrt\gamma)^{-1}$. By Lemma~\ref{lem:pibar_shape}, $\bar\pi$
is strictly increasing on $[0,r_{\bar\pi}]$ and strictly decreasing on
$[r_{\bar\pi},1]$.
\begin{itemize}
\item If $r_{\bar\pi}\ge r_1$, the linear class reaches $r_{\bar\pi}$, or
approaches it as $\alpha\to\infty$ when $r_{\bar\pi}=1$, and so earns the same
value as the threshold class.
\item If $r_{\bar\pi}<r_1$, the best level the linear class can hold is $r_1$,
and $\bar\pi(r_1)<\bar\pi(r_{\bar\pi})=\bar\pi(\tau^*)$.
\end{itemize}
By Lemma~\ref{lem:pibar_shape} and Proposition~\ref{prop:comparative_statics},
$r_{\bar\pi}$ equals $1$ for $M\le\gamma$, decreases continuously and strictly in
$M$ under \eqref{eq:interior_condition}, and equals $0$ beyond it. Hence
$r_{\bar\pi}<r_1$ exactly when $M$ exceeds the value at which $F(r_1,M)=0$.
Solving $F(r_1,M)=0$ from \eqref{eq:pibar_prime} for $M$ gives
\eqref{eq:M_max}. Finally, writing $a:=\sqrt\gamma$,
\[
M_{\max}-1=\frac{1-a+2p_ha-p_h(1+a)}{p_h(1+a)}=\frac{(1-a)(1-p_h)}{p_h(1+a)},
\]
which is \eqref{eq:Mmax_gap}, and it is positive because $a<1$ and $p_h<1$.
\end{proof}

In the picture of Figure~\ref{fig:target_point}, $M_{\max}$ is the tolerance at
which the target $z^\star$ crosses the diagonal. For $M\le M_{\max}$ the target lies
on or below the diagonal, and the optimal ray through it has slope at most one:
indeed $\alpha^*=r_{\bar\pi}/\hat u(r_{\bar\pi})$, which is
\eqref{eq:alpha_star_general}. For $M>M_{\max}$ the target lies above the
diagonal, and the best linear policy is the diagonal itself.

\begin{remark}[Equivalence in the fluid limit]
\label{rem:equivalence_global}
By Theorem~\ref{thm:fluid} and Remark~\ref{rem:threshold_attains}, in the fluid
limit the threshold rule at $\tau^*$ is optimal among all policies, and the
linear rule at $\alpha^*$ is as well whenever $\gamma<M\le M_{\max}$. The
shortfall for $M>M_{\max}$ comes only from the normalization $\alpha\ge1$. The
family $u=\min\{1,r/\alpha\}$ with $\alpha>0$ consists of rays of every slope,
reaches every interior target with $\alpha=r_{\bar\pi}/\hat u(r_{\bar\pi})$, and
attains the same value. What does not depend on the parametrization is that a
single operating level suffices. In the fluid limit, only the point at which a
policy crosses the zero-drift curve matters; at a fixed $\rho$, the shape of the
policy near $z^\star$ affects the payoff at order $\rho$,
which is the size of the correction allowed by Corollary~\ref{cor:finite_rho} (Section~\ref{sec:num_classes}).
\end{remark}

The economic content of Proposition~\ref{thm:equivalence} is that richer
trust-dependent behavior does not, by itself, generate long-run gains: every
class reduces to the choice of a single operating credibility level, and the
classes differ only in which levels they can hold. The loss asymmetry determines
whether that difference binds. The linear class cannot hold credibility below
$u^\sharp=(1+\sqrt\gamma)^{-1}$, the level at which the gap between recorded and
actual honesty is widest. When receivers are tolerant enough that the best level
lies below it ($M>M_{\max}$), a threshold rule is strictly better; for
$M\le M_{\max}$ the two classes are payoff-equivalent.

\section{Numerical Results}
\label{sec:numerical}

Section~\ref{sec:analysis} characterizes the best operating level exactly, but it
describes the dynamics only in the fluid limit, and its statements at a fixed
$\rho$ are bounds. This section asks what happens at the adjustment speeds a
platform would actually use. Section~\ref{sec:num_classes} compares the closed
forms with the unapproximated chain across receiver tolerance.
Sections~\ref{sec:num_transient} and~\ref{sec:num_dispersion} measure how fast each
policy reaches its operating level and how tightly it holds the score once there. Section~\ref{sec:num_premium} maps the value
of the credibility premium over the parameter space, and
Section~\ref{sec:num_robustness} relaxes three modeling choices and shows where the
main conclusion stops holding.

\paragraph{Setup.} We simulate
\eqref{eq:trust_update}--\eqref{eq:lambda_update} directly, without any of the
approximations of Section~\ref{sec:analysis}. Unless stated otherwise,
$p_h=\gamma=0.3$, $M=1$ ($c_l=1$, $c_o=0$) and $\rho=0.05$. At these values
$r_{\bar\pi}=\tau^*=0.769$, $\alpha^*=1.54$, the holding rate is
$\hat u(r_{\bar\pi})=0.5$, and $\bar\pi(r_{\bar\pi})=0.4225$.

We compare three stationary policies, each of which holds the score at
$r_{\bar\pi}$ in the fluid limit: the threshold policy at $\tau^*$, the linear
policy at $\alpha^*$, and the constant lying rate $u\equiv\hat u(r_{\bar\pi})$.
Payoffs are computed from the conditional purchase probability
$\E[a_t\mid\mathcal F_t]=\pi(r_t,u_t)$ rather than from realized sales; this has
the same expectation and a smaller variance. The exact optimum at a fixed $\rho$ is
computed by relative value iteration for the average-reward problem
\eqref{eq:sender_objective}, on $2001$ trust states and $21$ lying rates, with
linear interpolation of the value function at the two post-update states
$(1-\rho)r+\rho$ and $(1-\rho)r$. Horizons, burn-in periods and numbers of
replications are given in each caption. Error bars are $95\%$ confidence
intervals. Seeds are fixed per figure, and the accompanying code reproduces every
figure in this section.

\subsection{Closed Forms Across Receiver Tolerance}
\label{sec:num_classes}

For twelve values of $M\in[0.4,5]$ we search over the parameters of each class:
the cutoff $\tau$ in steps of $0.01$, the linear parameter $\alpha\ge1$ in steps of
$0.05$, and $(s,S)$ bands with $s\le S$ in steps of $0.02$. We then re-evaluate the
selected parameter with independent replications, so that the reported payoffs
are free of selection bias. Figure~\ref{fig:num_classes} summarizes the outcome;
it replaces the separate comparisons in $c_l$ and $c_o$, which by
Remark~\ref{rem:logodds} matter only through $M$.

The selected operating levels track $r_{\bar\pi}(M)$ of Lemma~\ref{lem:pibar_shape}
over the whole range (panel~(a)). The best cutoff lies within about $0.02$ of $r_{\bar\pi}$, and the offset has a
simple cause. Under a threshold policy the score is confined to the interval
$I_\tau$ of \eqref{eq:invariant_interval}, whose center is
$\tau-\rho(\tau-\tfrac12)$, so holding the score at $r_{\bar\pi}$ calls for the
cutoff $\tau\approx r_{\bar\pi}+\rho(r_{\bar\pi}-\tfrac12)$: slightly above
$r_{\bar\pi}$ when $r_{\bar\pi}>\tfrac12$ and slightly below it otherwise. At
$\rho=0.05$ this predicts $0.783$, $0.604$ and $0.363$ at $M=1$, $2$ and $4$,
against the selected $0.79$, $0.60$ and $0.35$. The linear class follows
$r_{\bar\pi}$ up to $M_{\max}=1.68$ and from $M=2$ on stays at $\alpha=1$, where its level
is $r^*_{lin}(1)=0.646$, exactly as Proposition~\ref{thm:equivalence} predicts.
The best $(s,S)$ band is degenerate for $M\in\{0.8,1,1.25,1.5\}$ and never wider
than $0.06$, which is the collapse of Theorem~\ref{thm:sS_collapse_threshold} on the
unapproximated chain.

Panel~(b) compares payoffs. The $(s,S)$ class is within $0.0006$ of the threshold
class at every $M$, and the exact optimum is within $0.0005$, so at this $\rho$ the threshold rule is within $0.2\%$ of the value-iteration
optimum, a difference no larger than the simulation and discretization error. Above
$M_{\max}$ the linear policy falls short by the amount the theory predicts: by
$0.014$ at $M=3$ (theory $0.014$) and by $0.047$ at $M=5$ (theory $0.048$). 

\begin{figure}[t]
\centering
\includegraphics[width=\linewidth]{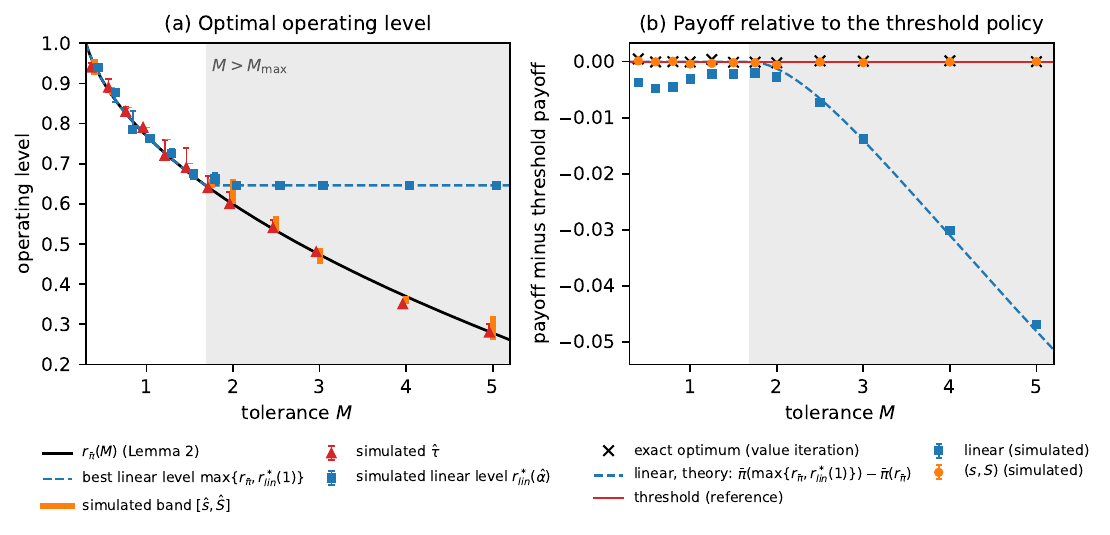}
\caption{Closed forms against the unapproximated chain across the tolerance $M$
($p_h=\gamma=0.3$, $\rho=0.05$). (a) Operating levels: $r_{\bar\pi}(M)$ of
Lemma~\ref{lem:pibar_shape} (solid), the best level reachable by the linear class
(dashed), the simulated best cutoff $\hat\tau$, the level $r^*_{lin}(\hat\alpha)$ of
the simulated best linear parameter, and the simulated best band $[\hat s,\hat S]$.
Error bars span the grid values whose payoff lies within one standard error of the
best. (b) Payoff of each class minus that of the threshold class, with the theory
gap of the linear class (dashed) and from value iteration optimum
($\times$). The shaded region is $M>M_{\max}$. Grid search: $T=3000$ periods with
the first $600$ discarded, $r_0=0.5$, and $300$, $250$ and $120$ paths per grid
point for $\tau$, $\alpha$ and $(s,S)$; each selected parameter is re-evaluated with
$2000$ independent paths.}
\label{fig:num_classes}
\end{figure}

\subsection{Convergence and the Value of Initial Credibility}
\label{sec:num_transient}

Section~\ref{sec:analysis} is about where the score settles. How fast it gets
there differs sharply across the three policies. In the fluid model the
constant rate approaches $r_{\bar\pi}$ exponentially, at rate
$(1-p_h)\gamma/D(r_{\bar\pi})=0.455$ per unit of fluid time. The linear policy
pulls back faster near the target, at rate
$(1-p_h)\bigl[\gamma/D(r_{\bar\pi})+D(r_{\bar\pi})/\alpha^*\bigr]=0.665$. The
threshold policy is the fastest possible approach: below the cutoff it is honest
in every low state, and above it, it lies in every low state, so it reaches the
target in finite time.

Figure~\ref{fig:num_transient}(a) shows that the fluid model predicts the mean
paths of the chain closely at $\rho=0.05$. Table~\ref{tab:hitting} reports the
time to come within $0.03$ of $r_{\bar\pi}$. From a low score the threshold policy
arrives in about a third of the time the constant rate needs, and from a high
score in about a third as well. Its arrival time also varies far less across
paths: the $90$th percentile exceeds the median by $6$ periods, against $54$ for
the constant rate.

\begin{table}[t]
\centering
\caption{Periods needed to come within $0.03$ of $r_{\bar\pi}=0.769$
($\rho=0.05$, $4000$ paths): median, $90$th percentile, and the fluid-model
prediction.}
\label{tab:hitting}
\begin{tabular}{@{}lcccccc@{}}
\toprule
 & \multicolumn{3}{c}{from $r_0=0.1$} & \multicolumn{3}{c}{from $r_0=0.95$}\\
\cmidrule(lr){2-4}\cmidrule(lr){5-7}
Policy & median & $90\%$ & fluid & median & $90\%$ & fluid\\
\midrule
constant rate $\hat u(r_{\bar\pi})$ & 103 & 157 & 137 & 49 & 111 & 79\\
linear, $\alpha^*$                   & 67  & 103 & 86  & 39 & 82  & 56\\
threshold, $\tau^*$                  & 35  & 41  & 35  & 17 & 30  & 17\\
\bottomrule
\end{tabular}
\end{table}

Speed matters for the payoff during the approach, and the accounting identity
\eqref{eq:affine_identity} says by how much. Let
$B_\sigma(r_0):=\sum_{t\ge0}\bigl(\E[a_t]-\mathcal J(\sigma,r_0)\bigr)$ be the
sales a policy earns during the transient beyond its long-run rate. Integrating
\eqref{eq:affine_identity} along the fluid path gives
\begin{equation}
B_\sigma(r_0)\;\approx\;
\frac{K(r_{\bar\pi})-K(r_0)}{\rho}
\;-\;\frac1\rho\int_0^\infty\bigl[\bar\pi(r_{\bar\pi})-\bar\pi\bigl(r(t)\bigr)\bigr]\,dt .
\label{eq:transient_value}
\end{equation}
The first term is the value of the credibility stock spent or built on the way to
$r_{\bar\pi}$; it is the same for every policy, which is the content of
Remark~\ref{rem:no_cycles}. The second term is the cost of time spent away from
the best level. It is nonnegative, and since $\bar\pi$ is single-peaked
(Lemma~\ref{lem:pibar_shape}) it is smallest for the fastest approach.

Figure~\ref{fig:num_transient}(b) confirms both parts. Starting from $r_0=0.05$,
the stock term is $-9.7$ sales, and the threshold, linear and constant policies
earn $-11.2$, $-12.3$ and $-13.8$. Starting from $r_0=0.95$, the stock term is
$+7.0$, and the three policies earn $+7.2$, $+6.5$ and $+5.9$. The fluid
prediction \eqref{eq:transient_value} agrees with the simulation to within about
half a sale at every $r_0$, the remainder being of order one rather than $1/\rho$.
Initial credibility is therefore an asset with a well-defined value, and the
threshold policy is the one that loses least of it on the way to the best level.

\begin{figure}[t]
\centering
\includegraphics[width=\linewidth]{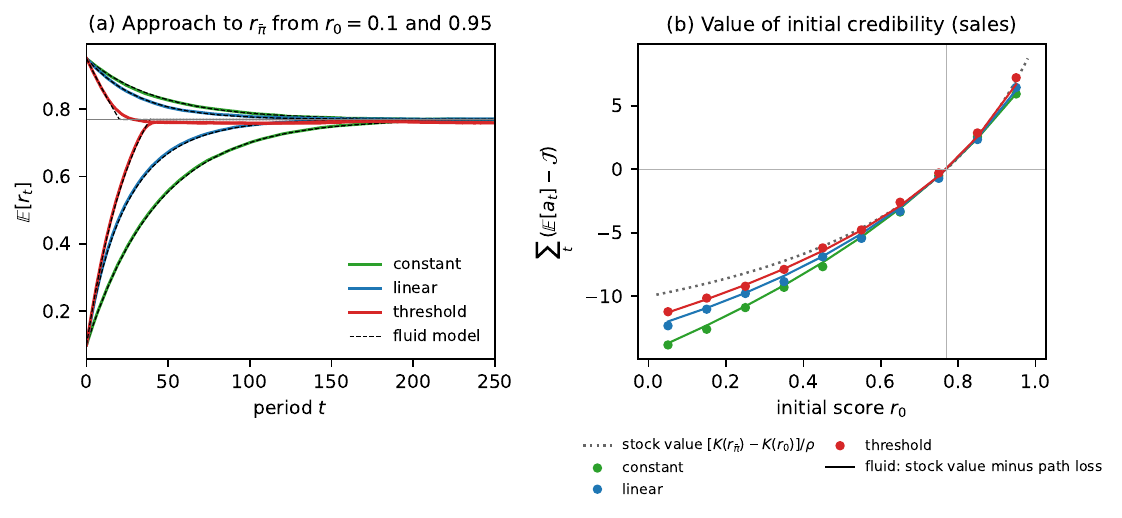}
\caption{Transient behavior at $\rho=0.05$. (a) Mean score paths from $r_0=0.1$ and
$r_0=0.95$ ($4000$ paths) with the fluid solutions of \eqref{eq:fluid_model}
(dashed). (b) Transient value $B_\sigma(r_0)$, the sales earned beyond the long-run
rate over the first $700$ periods ($4000$ paths per point), against the fluid prediction
\eqref{eq:transient_value} (lines) and the policy-independent stock value
$[K(r_{\bar\pi})-K(r_0)]/\rho$ (dotted).}
\label{fig:num_transient}
\end{figure}

\subsection{Stationary Fluctuations}
\label{sec:num_dispersion}

Once at the target, the policies differ in how tightly they hold the score. For a
smooth policy $\sigma$ the pull-back rate at $r_{\bar\pi}$ is
\[
\kappa_\sigma:=-\frac{d}{dr}\,b\bigl(r,\sigma(r)\bigr)\Big|_{r_{\bar\pi}}
=(1-p_h)\Bigl[\frac{\gamma}{D(r_{\bar\pi})}+D(r_{\bar\pi})\,\sigma'(r_{\bar\pi})\Bigr],
\]
and the per-period variance of $\lambda_t-r_t$ at the target is
\[
S:=(1-p_h)\Bigl[\bigl(1-\hat u(r_{\bar\pi})\bigr)(1-r_{\bar\pi})^2
+\hat u(r_{\bar\pi})\,\gamma\,r_{\bar\pi}^2\Bigr]=0.081 .
\]
Approximating the linearized dynamics by a diffusion then gives the stationary
variance $\rho S/(2\kappa_\sigma)$, which is of order $\rho$. The threshold policy
behaves differently. From any score below $\tau$, one step lands below
$(1-\rho)\tau+\rho$; from any score at or above $\tau$, a detected lie lands
above $(1-\rho)\tau$. The score is therefore eventually confined to
\begin{equation}
I_\tau:=\bigl[(1-\rho)\tau,\;(1-\rho)\tau+\rho\bigr],
\label{eq:invariant_interval}
\end{equation}
an interval of width $\rho$. Its variance is therefore at most $\rho^2/4$, of order
$\rho^2$. This interval was obtained in \cite{arcaklioglu2025parametric} for the
symmetric-loss case; the argument does not use the loss parameters.

Figure~\ref{fig:num_dispersion}(a) confirms both orders. For the constant and
linear policies the variance is proportional to $\rho$ and within $3\%$ of the
diffusion prediction for $\rho\le0.05$ ($13\%$ at $\rho=0.2$); the mean-square bound of
Theorem~\ref{thm:linear_fixed_point_concentration} is about four times larger than
the measured variance. For the threshold policy the variance is $0.09\rho^2$ at
every $\rho$, inside the bound $\rho^2/4$. At $\rho=0.05$ the standard deviations
are $0.068$ for the constant rate, $0.056$ for the linear policy and $0.015$ for
the threshold. The threshold policy holds the score $4.5$ times as tightly as the
constant rate, and the ratio grows like $\rho^{-1/2}$ as $\rho$ decreases. Panel~(b) shows the stationary
distributions. That of the threshold policy lives on $I_{\tau^*}$ and has a lattice
structure inherited from the two jump sizes. Its mean, $0.760$, lies slightly
below the cutoff because $I_{\tau^*}$ is centered at
$\tau^*-\rho(\tau^*-\tfrac12)=0.756$.

\begin{figure}[t]
\centering
\includegraphics[width=\linewidth]{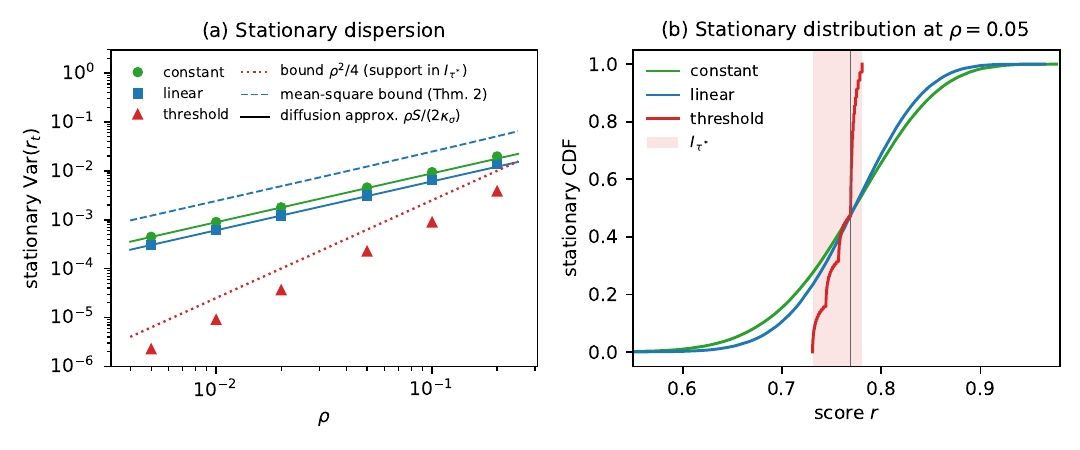}
\caption{Stationary fluctuations around $r_{\bar\pi}$. (a) Stationary variance of
the score against $\rho$, with the diffusion predictions $\rho S/(2\kappa_\sigma)$
for the constant and linear policies (solid), the bound $\rho^2/4$ implied by
\eqref{eq:invariant_interval} for the threshold policy (dotted), and the
mean-square bound of Theorem~\ref{thm:linear_fixed_point_concentration} for the
linear policy (dashed). (b) Stationary distribution functions at $\rho=0.05$; the
shaded band is $I_{\tau^*}$. Each point uses $6000$ paths started at $r_{\bar\pi}$,
a horizon $T=\max\{4000,200/\rho\}$ and a burn-in of $0.3T$.}
\label{fig:num_dispersion}
\end{figure}

\subsection{Value of the Credibility Premium}
\label{sec:num_premium}

Proposition~\ref{prop:debias} identifies the sender's entire gain from managing
credibility with the rent
$\bar\pi(r_{\bar\pi})-\max\{\bar\pi(0),\bar\pi(1)\}$.
Figure~\ref{fig:num_rent}(a) maps this rent over the detection probability $\gamma$
and the tolerance $M$ for $p_h=0.3$. It is zero outside the interior region
$\gamma<M<\overline M$ of \eqref{eq:interior_condition}, where the optimal level is a
corner. For every $\gamma$ it is largest at $M=1$, where always-honest and
always-deceptive behavior earn the same, $p_h$, so that neither corner is a good
outside option. It grows as detection weakens, reaching $0.43$ sales per period,
$144\%$ of the corner value, at $\gamma=0.02$. The orange curve is
$M_{\max}(\gamma)$ of \eqref{eq:M_max}; above it the linear class cannot collect the
full rent.

Panel~(b) takes the slice $M=1$ and compares it with the threshold policy
simulated at $\rho=0.01$. The rent is $0.335$, $0.253$, $0.170$, $0.123$, $0.067$
and $0.020$ at $\gamma=0.05$, $0.10$, $0.20$, $0.30$, $0.50$ and $0.80$, and it is
zero at $\gamma=1$, where by Lemma~\ref{lem:rest_point} the score equals the true
honesty rate. The simulated values exceed these by a small amount of order $\rho$, as Corollary~\ref{cor:finite_rho} allows. What the sender manages is therefore not trust as
such, but the wedge that imperfect detection opens between the recorded and the
actual honesty rate.

\begin{figure}[t]
\centering
\includegraphics[width=\linewidth]{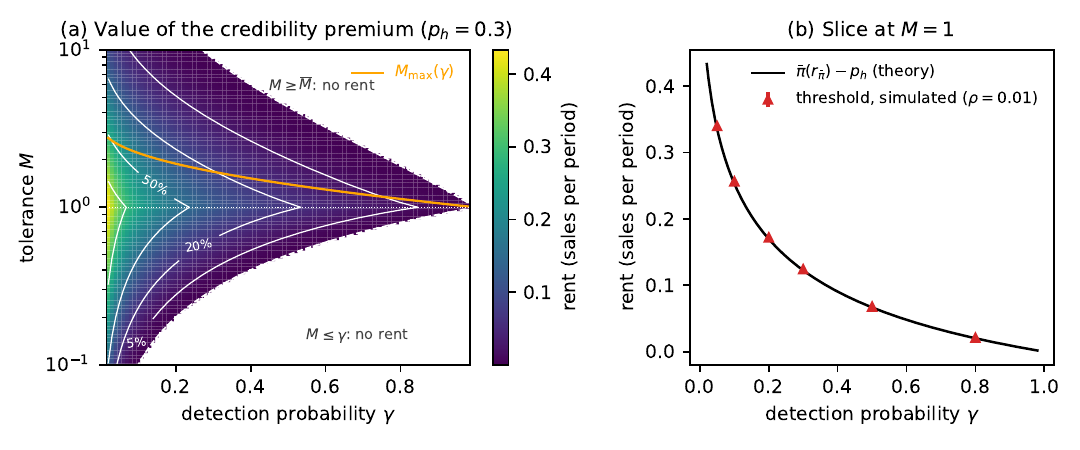}
\caption{The value of the credibility premium, $\bar\pi(r_{\bar\pi})-\max\{\bar\pi(0),
\bar\pi(1)\}$, at $p_h=0.3$. (a) The rent over $(\gamma,M)$, in sales per period;
white contours give it relative to the corner value $\max\{\bar\pi(0),\bar\pi(1)\}$,
and the orange curve is $M_{\max}(\gamma)$. The rent is zero outside
$\gamma<M<\overline M$. (b) The slice $M=1$: theory (line) and the threshold policy
simulated at $\rho=0.01$ ($2000$ paths, $T=20000$, burn-in $6000$).}
\label{fig:num_rent}
\end{figure}

\subsection{Robustness and Limits}
\label{sec:num_robustness}

\paragraph{Detection that requires a purchase.} In the baseline a lie can be
detected even without a sale (Remark~\ref{rem:detection}). If only purchased
lemons can be found out, the detection probability becomes $\gamma q(r_t)$.
Reward and drift remain affine in $u$, so the identity
\eqref{eq:affine_identity} holds with $D(r)$ replaced by
$D_q(r):=1-r+\gamma q(r)r\ge\gamma q(0)>0$, and Theorem~\ref{thm:fluid}
applies unchanged: a single operating level is still optimal in the fluid
limit, now at the maximizer $r_q$ of
$\bigl(p_h+(1-p_h)\hat u_q(r)\bigr)q(r)$ with $\hat u_q(r)=(1-r)/D_q(r)$.
Figure~\ref{fig:robust}(a,b) compare the two rules as $c_l$ varies. The
sender's payoff is uniformly higher under purchase-dependent detection, by
$7$--$14\%$, because detection is always weaker. The optimal cutoff still
increases in $c_l$ but more slowly, and the two curves cross: for cautious
buyers (large $c_l$) purchases are rare, detection is much weaker, and the
sender holds a lower level ($r_q=0.84$ against $0.90$ at $c_l=2$); for
tolerant buyers detection rises steeply with the score, which makes high
levels relatively cheaper to hold, and the cutoff is slightly higher.

\paragraph{Switching costs.} Figure~\ref{fig:robust}(c) charges the sender a fixed cost $\eta$ per change of mode in the $(s,S)$ class and reports the optimized band width. The band stays degenerate for $\eta\le0.005$, widens to $0.05$ at $\eta=0.01$, and reaches the upper limit of the search, $0.10$, for $\eta\ge0.02$; wider bands were not searched, so the plateau marks the limit of the grid rather than a preference for width $0.10$. This supports the conjecture stated
after Theorem~\ref{thm:sS_collapse_threshold}: within the bang-bang classes it is
the absence of a switching cost that makes wide bands worthless. A sender who can randomize, however, avoids the cost altogether by lying at the constant rate $\hat u(r_{\bar\pi})$, which never changes mode, so switching costs restore bands
only for senders restricted to pure honest and pure exploiting modes.

\begin{figure}[t]
\centering
\includegraphics[width=\linewidth]{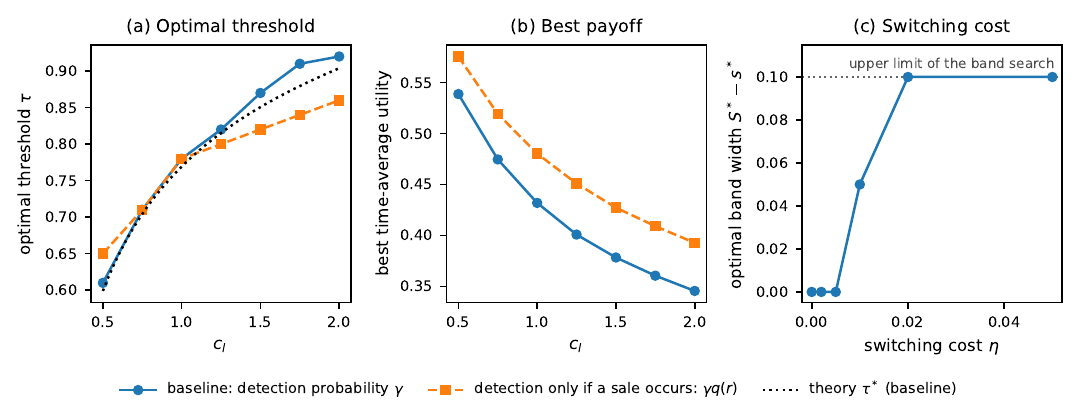}
\caption{Robustness checks. (a) Optimal cutoff and (b) best payoff under the baseline detection rule (probability $\gamma$) and under detection conditional on a sale (probability $\gamma q(r)$), as $c_l$ varies with $c_o=0$; the dotted line in (a) is $\tau^*$ of Theorem~\ref{thm:optimal_threshold}. (c) Optimized
band width $S^*-s^*$ in the $(s,S)$ class as a function of a fixed per-switch cost $\eta$, with $c_l=1$, $c_o=0$. Bands up to width $0.10$ were searched; for $\eta\ge0.02$ the optimum sits at that upper limit.}
\label{fig:robust}
\end{figure}

\paragraph{Where cycles pay: seasonal demand.} Theorem~\ref{thm:fluid} uses that
the environment does not change over time. Suppose instead that the number of
buyers varies seasonally, so that a sale in period $t$ is worth $w_t=1+a\sin(2\pi t/P)$. Multiplying \eqref{eq:affine_identity} by $w$ gives
$w\pi=w\bar\pi+\frac{d}{dt}(wK)-\dot w K$, and the last term no longer telescopes: credibility is worth more just before demand peaks than just after. A seller can then gain by building credibility while demand rises and spending it at the peak, which is a build-and-exploit cycle.

To measure the gain we compare the best constant cutoff with periodic cutoffs
$\tau_t=\tau^*+\epsilon\cos(2\pi t/P+\phi)$, searching over $\epsilon$ and $\phi$ and
using the same random numbers for every candidate.
Figure~\ref{fig:num_seasonal}(a) reports the gain. It is negligible for slow seasons
($P=1000$ periods), but reaches $1.3\%$ for $P=400$ and $5.4\%$ for $P=200$ at
amplitude $a=0.5$, and it grows with the amplitude. In every case with a
significant gain the best phase is $\phi=0$: as panel~(b) shows, the score peaks a
quarter of a season before demand does and is drawn down while demand is
highest. The single-operating-level conclusion therefore rests on stationarity of
the environment, and a seasonal seller should run a cycle synchronized with
demand.

\begin{figure}[t]
\centering
\includegraphics[width=\linewidth]{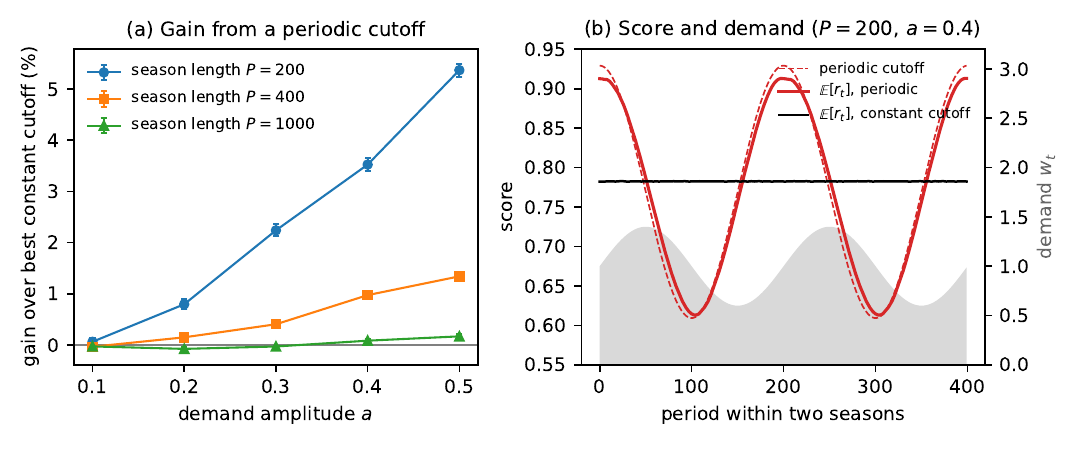}
\caption{Seasonal demand $w_t=1+a\sin(2\pi t/P)$ at $\rho=0.05$. (a) Gain of the
best periodic cutoff $\tau^*+\epsilon\cos(2\pi t/P+\phi)$ over the best constant
cutoff, in percent of demand-weighted sales, against the amplitude $a$ for three
season lengths $P$. The candidates share common random numbers, and the two
selected policies are re-evaluated with $1500$ independent paths over six seasons.
(b) One example ($P=200$, $a=0.4$): demand (shaded), the periodic cutoff (dashed),
and the mean score under the periodic and the constant cutoff.}
\label{fig:num_seasonal}
\end{figure}

\section{Conclusion}
\label{sec:conclusion}

We studied a long-run sender who manages a public credibility score in front
of myopic receivers with asymmetric losses. Receiver behavior depends on the
two losses only through the tolerance $M$, and because only detected
misrepresentations enter the record, the score rests above the sender's true
honesty rate by a credibility premium we characterize in closed form
(Lemma~\ref{lem:rest_point}). For the linear, threshold and $(s,S)$ classes we
derived the optimal parameters in closed form, showed that every nondegenerate
band is weakly dominated by a cutoff, and showed that all three classes earn
the same function $\bar\pi$ of the level they hold, so the threshold class
strictly dominates exactly when $M>M_{\max}$
(Proposition~\ref{thm:equivalence}). Treating credibility as a stock with a
path-independent shadow price, we showed that in the fluid limit no policy
earns more than operating at the single level $r_{\bar\pi}$
(Theorem~\ref{thm:fluid}) and that the sender's entire gain from managing
credibility is the rent created by censored detection
(Proposition~\ref{prop:debias}). At a fixed adjustment speed every policy is bounded by $\max_r\bar\pi(r)+O(\rho)$ (Corollary~\ref{cor:finite_rho}), and in the simulations of Section~\ref{sec:numerical} the threshold policy is within the simulation error of the value-iteration optimum.

The simulations also mark where the single-level conclusion ends: bands pay
once switching modes carries a fixed cost, and build-and-exploit cycles pay
when demand is seasonal. Two assumptions delimit the rest. Receivers read the
score at face value, motivated by the documented inflation of platform feedback
\citep{dellarocas2008silence, filippas2022} rather than derived from
equilibrium, and Proposition~\ref{prop:debias} shows that receivers who invert
the censoring remove the sender's rent. The sender maximizes an undiscounted
average with a parameter fixed once; the transient of an impatient sender
starting from a high score is a natural next step.

One extension is  planned.  We plan to replace the binary quality state by a continuous one, with a continuum of reports and a purchase rule that depends on the distance between the reported and the realized quality; the drift--reward identity \eqref{eq:affine_identity} is the natural starting point for asking whether a single operating level remains optimal when misrepresentation has a degree as well as a rate.

\bibliographystyle{ieeetr}
\bibliography{journal}

\begin{thebibliography}{10}

\bibitem{kamenica2011bayesian}
E.~Kamenica and M.~Gentzkow, ``Bayesian persuasion,'' {\em American Economic Review}, vol.~101, pp.~2590--2615, 2011.

\bibitem{crawford1982strategic}
V.~P. Crawford, ``Strategic information transmission,'' {\em Econometrica}, vol.~50, pp.~1431--1451, 1982.

\bibitem{akyol2016information}
E.~Akyol, C.~Langbort, and T.~Ba{\c{s}}ar, ``Information-theoretic approach to strategic communication as a hierarchical game,'' {\em Proceedings of the IEEE}, vol.~105, no.~2, pp.~205--218, 2016.

\bibitem{little2023bayesian}
A.~T. Little, ``Bayesian explanations for persuasion,'' {\em Journal of Theoretical Politics}, vol.~35, pp.~147--181, 2023.

\bibitem{sayin2021bayesian}
M.~O. Sayin and T.~Ba{\c{s}}ar, ``Bayesian persuasion with state-dependent quadratic cost measures,'' {\em IEEE Transactions on Automatic Control}, vol.~67, no.~3, pp.~1241--1252, 2021.

\bibitem{velicheti2023strategic}
R.~K. Velicheti, M.~Bastopcu, and T.~Ba{\c{s}}ar, ``Strategic information design in quadratic multidimensional persuasion games with two senders,'' in {\em 2023 American Control Conference (ACC)}, pp.~1716--1722, IEEE, 2023.

\bibitem{massicot2025almost}
O.~Massicot and C.~Langbort, ``Almost-bayesian quadratic persuasion,'' {\em IEEE Transactions on Automatic Control}, 2025.

\bibitem{honryo2018dynamic}
T.~Honryo, ``Dynamic persuasion,'' {\em Journal of Economic Theory}, vol.~178, pp.~36--58, 2018.

\bibitem{ball2023dynamic}
I.~Ball, ``Dynamic information provision: Rewarding the past and guiding the future,'' {\em Econometrica}, vol.~91, pp.~1171--1565, Jul 2023.

\bibitem{ely2017beeps}
J.~C. Ely, ``Beeps,'' {\em American Economic Review}, vol.~107, pp.~31--53, 2017.

\bibitem{best2024persuasion}
J.~Best and D.~Quigley, ``Persuasion for the long run,'' {\em Journal of Political Economy}, vol.~132, no.~5, pp.~1740--1791, 2024.

\bibitem{mathevet2018reputation}
L.~Mathevet, D.~Pearce, and E.~Stacchetti, ``Reputation and information design,'' working paper, University of Arizona and New York University, 2026.
\newblock Version of May 6, 2026.

\bibitem{barisaac2003reputation}
H.~Bar-Isaac, ``Reputation and survival: Learning in a dynamic signalling model,'' {\em Review of Economic Studies}, vol.~70, pp.~231--251, 2003.

\bibitem{manapat2013information}
M.~L. Manapat, M.~A. Nowak, and D.~G. Rand, ``Information, irrationality, and the evolution of trust,'' {\em Journal of Economic Behavior \& Organization}, vol.~90, no.~S, pp.~57--75, 2013.

\bibitem{fudenberg2022honesty}
D.~Fudenberg, Y.~Gao, and H.~Pei, ``A reputation for honesty,'' {\em Journal of Economic Theory}, vol.~204, p.~105508, 2022.

\bibitem{jullien2014new}
B.~Jullien and I.-U. Park, ``New, like new, or very good? reputation and credibility,'' {\em The Review of Economic Studies}, vol.~81, no.~4, pp.~1543--1574, 2014.

\bibitem{liu2014limited}
Q.~Liu and A.~Skrzypacz, ``Limited records and reputation bubbles,'' {\em Journal of Economic Theory}, vol.~151, pp.~2--29, 2014.

\bibitem{cripps2004imperfect}
M.~W. Cripps, G.~J. Mailath, and L.~Samuelson, ``Imperfect monitoring and impermanent reputations,'' {\em Econometrica}, vol.~72, no.~2, pp.~407--432, 2004.

\bibitem{dellarocas2008silence}
C.~Dellarocas and C.~A. Wood, ``The sound of silence in online feedback: Estimating trading risks in the presence of reporting bias,'' {\em Management Science}, vol.~54, no.~3, pp.~460--476, 2008.

\bibitem{filippas2022}
A.~Filippas, J.~J. Horton, and J.~M. Golden, ``Reputation inflation,'' {\em Marketing Science}, vol.~41, no.~4, pp.~733--745, 2022.

\bibitem{velicheti2024learning}
R.~K. Velicheti, M.~Bastopcu, S.~R. Etesami, and T.~Ba{\c{s}}ar, ``Learning how to strategically disclose information,'' in {\em 2024 American Control Conference (ACC)}, pp.~1604--1609, IEEE, 2024.

\bibitem{farokhi2016estimation}
F.~Farokhi, A.~M. Teixeira, and C.~Langbort, ``Estimation with strategic sensors,'' {\em IEEE Transactions on Automatic Control}, vol.~62, no.~2, pp.~724--739, 2016.

\bibitem{bacsar2024incentive}
T.~Ba{\c{s}}ar, ``Incentive designs under disparate probabilistic outlooks across decision makers,'' {\em IFAC-PapersOnLine}, vol.~58, no.~30, pp.~7--12, 2024.

\bibitem{cho1987signaling}
I.-K. Cho and D.~M. Kreps, ``Signaling games and stable equilibria,'' {\em The Quarterly Journal of Economics}, vol.~102, pp.~179--221, 1987.

\bibitem{saritas2017quadratic}
S.~Sar{\i}ta\c{s}, S.~Y{\"u}ksel, and S.~Gezici, ``Quadratic multi-dimensional signaling games and affine equilibria,'' {\em IEEE Transactions on Automatic Control}, vol.~62, no.~2, pp.~605--619, 2017.

\bibitem{sayin2019hierarchical}
M.~O. Sayin, E.~Akyol, and T.~Ba\c{s}ar, ``Hierarchical multistage {G}aussian signaling games in noncooperative communication and control systems,'' {\em Automatica}, vol.~107, pp.~9--20, 2019.

\bibitem{saritas2020dynamic}
S.~Sar{\i}ta\c{s}, S.~Y{\"u}ksel, and S.~Gezici, ``Dynamic signaling games with quadratic criteria under {N}ash and {S}tackelberg equilibria,'' {\em Automatica}, vol.~115, p.~108883, 2020.

\bibitem{doval2020sequential}
L.~Doval and J.~C. Ely, ``Sequential information design,'' {\em Econometrica}, vol.~88, no.~6, pp.~2575--2608, 2020.

\bibitem{farhadi2022dynamic}
F.~Farhadi and D.~Teneketzis, ``Dynamic information design: A simple problem on optimal sequential information disclosure,'' {\em Dynamic Games and Applications}, vol.~12, no.~2, pp.~443--484, 2022.

\bibitem{wu2022markov}
J.~Wu, Z.~Zhang, Z.~Feng, Z.~Wang, Z.~Yang, M.~I. Jordan, and H.~Xu, ``Sequential information design: Markov persuasion process and its efficient reinforcement learning,'' in {\em Proceedings of the 23rd ACM Conference on Economics and Computation (EC)}, 2022.

\bibitem{renault2013dynamic}
J.~Renault, E.~Solan, and N.~Vieille, ``Dynamic sender–receiver games,'' {\em Journal of Economic Theory}, vol.~148, pp.~502--534, 2013.

\bibitem{basar1999dynamic}
T.~Ba\c{s}ar and G.~J. Olsder, {\em Dynamic Noncooperative Game Theory}.
\newblock Philadelphia, PA: SIAM, 2nd~ed., 1999.

\bibitem{ouyang2017dynamic}
Y.~Ouyang, H.~Tavafoghi, and D.~Teneketzis, ``Dynamic games with asymmetric information: Common information based perfect {B}ayesian equilibria and sequential decomposition,'' {\em IEEE Transactions on Automatic Control}, vol.~62, no.~1, pp.~222--237, 2017.

\bibitem{ottaviani2006reputational}
M.~Ottaviani and P.~N. S{\o}rensen, ``Reputational cheap talk,'' {\em The RAND Journal of Economics}, vol.~37, pp.~155--175, Mar 2006.

\bibitem{backus2020dont}
M.~Backus and A.~T. Little, ``{I Don’t Know},'' {\em American Political Science Review}, vol.~114, pp.~724--743, Aug 2020.

\bibitem{gentzkow2006media}
M.~Gentzkow and J.~M. Shapiro, ``Media bias and reputation,'' {\em Journal of Political Economy}, vol.~114, no.~2, pp.~280--316, 2006.

\bibitem{prendergast1993theory}
C.~Prendergast, ``A theory of 'yes men','' {\em American Economic Review}, vol.~83, no.~4, pp.~757--770, 1993.

\bibitem{morris2001political}
S.~Morris, ``Political correctness,'' {\em Journal of Political Economy}, vol.~109, pp.~231--265, Apr 2001.

\bibitem{ma2021reputation}
Q.~Ma, J.~Huang, T.~Basar, J.~Liu, and X.~Chen, ``Reputation and pricing dynamics in online markets,'' {\em IEEE/ACM Transactions on Networking}, vol.~29, no.~4, pp.~1745--1759, 2021.

\bibitem{zhu2018stability}
Y.~Zhu and K.~Savla, ``On the stability of optimal bayesian persuasion strategy under a mistrust dynamics in routing games,'' in {\em 2018 56th Annual Allerton Conference on Communication, Control, and Computing (Allerton)}, pp.~92--99, IEEE, 2018.

\bibitem{zhu2019routing}
Y.~Zhu and K.~Savla, ``On routing drivers through persuasion in the long run,'' in {\em 2019 IEEE 58th Conference on Decision and Control (CDC)}, pp.~4091--4096, IEEE, 2019.

\bibitem{zhu2022information}
Y.~Zhu and K.~Savla, ``Information design in nonatomic routing games with partial participation: Computation and properties,'' {\em IEEE Transactions on Control of Network Systems}, vol.~9, no.~2, pp.~613--624, 2022.

\bibitem{gurdal2014truth}
M.~Y. Gurdal, A.~Ozdogan, and I.~Saglam, ``Truth-telling and trust in sender–receiver games with intervention: An experimental study,'' {\em Review of Economic Design}, vol.~18, pp.~83--103, 2014.

\bibitem{deori2024informativeness}
R.~Deori and A.~A. Kulkarni, ``Informativeness and trust in bayesian persuasion,'' {\em arXiv preprint arXiv:2408.13822}, Aug 2024.

\bibitem{vonsydow2019temporal}
M.~von Sydow, C.~Merdes, and U.~Hahn, ``The temporal dynamics of belief-based updating of epistemic trust: Light at the end of the tunnel?,'' {\em arXiv preprint arXiv:1912.13380}, Dec 2019.

\bibitem{11107502}
D.~H. Mangalindan, A.~R. Hota, and V.~Srivastava, ``On logit dynamics for multiplayer trust game,'' in {\em 2025 American Control Conference (ACC)}, pp.~4135--4140, 2025.

\bibitem{ekmekci2011sustainable}
M.~Ekmekci, ``Sustainable reputations with rating systems,'' {\em Journal of Economic Theory}, vol.~146, no.~2, pp.~479--503, 2011.

\bibitem{board2013reputation}
S.~Board and M.~Meyer-ter Vehn, ``Reputation for quality,'' {\em Econometrica}, vol.~81, no.~6, pp.~2381--2462, 2013.

\bibitem{djawadi2018maintaining}
B.~Mir~Djawadi, R.~Fahr, C.-J. Haake, and S.~Recker, ``Maintaining vs. milking good reputation when customer feedback is inaccurate,'' {\em PLoS ONE}, vol.~13, no.~11, p.~e0207172, 2018.

\bibitem{ko2025persuasion}
C.~Y. Ko, ``Persuasion in the long run: When history matters,'' {\em arXiv preprint arXiv:2508.01662}, 2025.

\bibitem{lukyanov2025mutual}
G.~Lukyanov, ``Mutual reputation and trust in a repeated sender--receiver game,'' {\em arXiv preprint arXiv:2509.04035}, 2025.

\bibitem{luo2025marginal}
D.~Luo and A.~Wolitzky, ``Marginal reputation,'' {\em Econometrica}, vol.~93, no.~6, pp.~2007--2042, 2025.

\bibitem{dellarocas2005reputation}
C.~Dellarocas, ``Reputation mechanism design in online trading environments with pure moral hazard,'' {\em Information Systems Research}, vol.~16, no.~2, pp.~209--230, 2005.

\bibitem{aperjis2010optimal}
C.~Aperjis and R.~Johari, ``Optimal windows for aggregating ratings in electronic marketplaces,'' {\em Management Science}, vol.~56, no.~5, pp.~864--880, 2010.

\bibitem{cabral2010dynamics}
L.~Cabral and A.~Horta{\c{c}}su, ``The dynamics of seller reputation: Evidence from ebay,'' {\em The Journal of Industrial Economics}, vol.~58, no.~1, pp.~54--78, 2010.

\bibitem{kushner2003stochastic}
H.~J. Kushner and G.~G. Yin, {\em Stochastic Approximation and Recursive Algorithms and Applications}.
\newblock Springer, 2nd~ed., 2003.

\bibitem{borkar2008stochastic}
V.~S. Borkar, {\em Stochastic Approximation: A Dynamical Systems Viewpoint}.
\newblock Cambridge University Press and Hindustan Book Agency, 2008.

\bibitem{scarf1960optimality}
H.~Scarf, ``The optimality of {$(S,s)$} policies in the dynamic inventory problem,'' in {\em Mathematical Methods in the Social Sciences, 1959} (K.~J. Arrow, S.~Karlin, and P.~Suppes, eds.), pp.~196--202, Stanford University Press, 1960.

\bibitem{dziuda2018communication}
W.~Dziuda and C.~Salas, ``Communication with detectable deceit,'' tech. rep., SSRN, Mar 2018.
\newblock Available: \url{https://ssrn.com/abstract=3234695}.

\bibitem{titova2021persuasion}
M.~Titova, ``Persuasion with verifiable information.'' Unpublished manuscript, latest version, 2021.

\bibitem{drakopoulos2021persuading}
K.~Drakopoulos, S.~Jain, and R.~Randhawa, ``Persuading customers to buy early: The value of personalized information provisioning,'' {\em Management Science}, vol.~67, no.~2, pp.~828--853, 2021.

\bibitem{bonatti2023coordination}
A.~Bonatti, M.~Dahleh, T.~Horel, and A.~Nouripour, ``Coordination via selling information,'' {\em arXiv preprint arXiv:2302.12223}, 2023.

\bibitem{wu2021value}
M.~Wu, S.~Amin, and A.~E. Ozdaglar, ``Value of information in bayesian routing games,'' {\em Operations Research}, vol.~69, no.~1, pp.~148--163, 2021.

\bibitem{arcaklioglu2025parametric}
A.~E. Arcaklioglu and C.~Eksin, ``Persuasion under dynamic trust: Parametric sender behavior,'' in {\em 2026 IEEE 65th Conference on Decision and Control (CDC)}, IEEE, 2026.
\newblock To appear.

\bibitem{xu2014onlineexclusive}
L.~Xu and X.~Xu, ``A decision problem of hotels on setting online-exclusive-rooms in e-commerce age,'' {\em Journal of Service Science and Management}, vol.~7, no.~4, pp.~293--302, 2014.

\bibitem{tadelis2016}
S.~Tadelis, ``Reputation and feedback systems in online platform markets,'' {\em Annual Review of Economics}, vol.~8, no.~1, pp.~321--340, 2016.

\end{thebibliography}

\appendix
\section{Foundation of the Purchase Rule}\label{sec:receiver_foundation}

We now justify \eqref{eq:probabilistic_purchase}, starting from the behavior of a
single expected-utility-maximizing receiver.

\begin{proposition}[Receiver's purchase cutoff]
\label{prop:receiver_threshold}
Under receiver's expected-utility maximization, the receiver purchases after
observing a high signal if and only if
\begin{equation}
    \mu \geq (M+1)^{-1}.
\end{equation}
Equivalently, persuasion requires trust to satisfy
\begin{equation}
    r_t\ge r_B^+,
\qquad
r_B^+:=
\max\left\{0,\frac{(M+1)^{-1}-p_h}{(M+1)^{-1}(1-p_h)}\right\}.\label{r_b_def}
\end{equation}
\end{proposition}

\begin{proof}
Given posterior belief $\mu$, the expected payoff from purchasing and not
purchasing are
\begin{equation}
    U_R(1\mid \mu):=\mu-(1-\mu)c_l, \quad  U_R(0\mid \mu):=-c_o\mu.
\end{equation}
Thus, the receiver purchases if and only if
\begin{equation}
    \mu-(1-\mu)c_l\ge -c_o\mu, \text{ or equivalently: }  \mu\ge \frac{c_l}{1+c_l+c_o}=(M+1)^{-1}.
\end{equation}
Since $\mu$ is increasing in
$r_t$, this is equivalent to \eqref{r_b_def}. If $(M+1)^{-1}\le p_h$, even $r_t=0$ is sufficient for persuasion, so the relevant
boundary is truncated at zero.
\end{proof}

Proposition~\ref{prop:receiver_threshold} shows that receiver-side loss
asymmetries enter the sender's problem through the persuasion boundary $r_B^+$.
Increasing $c_l$ raises the required trust level, whereas increasing
$c_o$ lowers it. Both effects operate through the tolerance
$M$ of \eqref{eq:M_def}: the purchase condition $\mu\ge(M+1)^{-1}$ can be written as
$M\ge(1-\mu)/\mu$, so the receiver buys exactly when its tolerance exceeds the
odds against high quality, and $(M+1)^{-1}$ is a decreasing transform of $M$.

The probabilistic purchase rule is not a behavioral departure from Proposition~\ref{prop:receiver_threshold}. We interpret this rule as an aggregate outcome over a heterogeneous
population of myopic receivers: each receiver applies the optimal cutoff of
Proposition~\ref{prop:receiver_threshold}, and the smooth dependence of sales on
trust arises from heterogeneity in receiver losses rather than from
randomization by any individual.

\begin{lemma}[Micro-foundation of the purchase rule]
\label{lem:microfoundation}
Suppose receivers are heterogeneous: receiver $i$ has losses
$(c^i_{l},c^i_{o})$ and tolerance
$M_i:=(1+c^i_{o})/c^i_{l}$,
drawn i.i.d. across receivers from the Lomax (Pareto Type~II) distribution with
shape $1$ and scale $M$ of \eqref{eq:M_def}, that is,
\begin{equation}
\Prob\left(M_i>x\right)=\left(1+\frac{x}{M}\right)^{-1},
\qquad x\ge 0 .
\label{eq:lomax}
\end{equation}
If every receiver best responds as in Proposition~\ref{prop:receiver_threshold},
then the population purchase share after a high signal is exactly
\eqref{eq:probabilistic_purchase}, where $\mu$ is given by
\eqref{eq:posterior_H}. Moreover, $M$ is the median tolerance in the population.
\end{lemma}

\begin{proof}
By the argument in Proposition~\ref{prop:receiver_threshold}, receiver $i$
purchases after a high signal if and only if
\begin{equation}
M_i\;\ge\;x,
\qquad x:=\frac{1-\mu}{\mu},
\label{eq:tolerance_cutoff}
\end{equation}
so a receiver buys precisely when its tolerance exceeds the odds against high
quality. Since $M_i$ is continuously distributed, ties occur with probability
zero, and the population purchase share equals $\Prob(M_i\ge x)=\Prob(M_i>x)$.
Applying \eqref{eq:lomax},
\begin{equation}
\Prob\left(M_i>x\right)
=\frac{M}{M+x}
=\frac{M}{M+\frac{1-\mu}{\mu}}
=\frac{\mu M}{\mu M+(1-\mu)} .
\end{equation}
Substituting $M=(1+c_o)/c_l$ and multiplying numerator and
denominator by $c_l$ yields \eqref{eq:probabilistic_purchase}. Finally,
setting $\Prob(M_i>x)=\tfrac12$ in \eqref{eq:lomax} gives $x=M$, so $M$ is the
median of $M_i$.
\end{proof}

\begin{remark}[Role of the tolerance distribution]
\label{rem:tolerance_dist}
The shape-$1$ Lomax law has an equivalent construction: $M_i=M\,U_i/(1-U_i)$ with
$U_i\sim\mathrm{Unif}(0,1)$, i.e.\ $\log M_i=\log M+\varepsilon_i$ with
$\varepsilon_i$ standard logistic, so $M_i$ is log-logistic with median $M$ and
the purchase decision is a logit choice in log-tolerance. Any tolerance
distribution with survival function $\bar F$ induces the purchase share
$q(\mu)=\bar F\bigl((1-\mu)/\mu\bigr)$, which is nondecreasing in $\mu$, and
\eqref{eq:lomax} is the unique choice that yields the multiplicative-odds rule
\eqref{eq:logodds} for every odds level the model generates,
$x\in[0,(1-p_h)/p_h]$. The distribution above $(1-p_h)/p_h$ is irrelevant,
because those receivers buy after every high signal; in particular, the heavy
right tail and the infinite mean $\E[M_i]=\infty$ never enter the analysis,
and $M$ enters as the median, which is the natural summary of a population of
threshold rules. The results differ in how much they use \eqref{eq:lomax}.
Lemma~\ref{lem:rest_point} and Thm.~\ref{thm:linear_fixed_point_concentration}
do not involve $q$. Theorems~\ref{thm:fluid} and
\ref{thm:sS_collapse_threshold}, and the equality of the three objectives in
Proposition~\ref{thm:equivalence}, use only that $q$ is a bounded, continuous
function of $r$, with $\bar\pi$ in \eqref{eq:common_value} evaluated at the
corresponding $q$. Corollary~\ref{cor:finite_rho} and \eqref{eq:payoff_error}
hold whenever $q$ is Lipschitz, with constants that depend on the
distribution. The closed forms \eqref{eq:r_pibar},
\eqref{eq:tau_star} and \eqref{eq:M_max}, the interior condition
\eqref{eq:interior_condition},
Propositions~\ref{prop:comparative_statics}--\ref{prop:loss_limits}, and the
corner characterization \eqref{eq:debias_max} in
Proposition~\ref{prop:receiver_models} use \eqref{eq:lomax}.
\end{remark}

Proposition~\ref{prop:receiver_threshold} should therefore be read as the
individual building block of Lemma~\ref{lem:microfoundation}: the persuasion
boundary $r_B^+$ describes when a \emph{single} receiver with losses
$(c_l,c_o)$ is persuaded, while the analysis from
Section~\ref{sec:analysis} onwards uses the population share $q$. In the heterogeneous model $r_B^+$ retains a meaning: since
$q(r)\ge\tfrac12$ if and only if $\mu\ge (M+1)^{-1}$, the boundary $r_B^+$ is
the score at which the median receiver is persuaded.

\section{Explicit form of $F(r,M)$}
\begin{equation}
    F(r,M):=(1-\gamma)(1-p_h\gamma)r^2-2(1-\gamma)r+1-\gamma+p_h\gamma-p_h\gamma M . \label{F_long_form}
\end{equation}

\section{Log-odds form}
\begin{remark}
\label{rem:logodds}
Writing $q$ for the purchase probability in \eqref{eq:probabilistic_purchase},
\begin{equation}
\frac{q}{1-q}=M\cdot\frac{\mu}{1-\mu},
\qquad\text{equivalently}\qquad
\mathrm{logit}\,q=\log M+\mathrm{logit}\,\mu ,
\label{eq:logodds}
\end{equation}
so the odds of purchase are the odds of high quality scaled by the tolerance $M$. This makes transparent that receiver preferences enter the sender's
problem only through $M$: the loss parameters shift the purchase log-odds by a
constant. It also shows that $q=\mu$ exactly when $M=1$, that is, when
$c_l=1+c_o$, a useful special case.
\end{remark}

\section{Proof of Lemma~\ref{lem:pibar_shape}}
\begin{proof}
Write $\nu(r):=1-(1-p_h\gamma)r$ and $\omega(r):=p_hM+(1-p_h)(1-r)$. Dividing
\eqref{eq:probabilistic_purchase} by $c_l$ gives $q=p_hM/\omega$, and
$p_h+(1-p_h)\hat u=\bigl(p_hD+(1-p_h)(1-r)\bigr)/D=\nu/D$, so $\bar\pi=\nu q/D$.
Since $\nu'=-(1-p_h\gamma)$, $D'=-(1-\gamma)$, $q'/q=(1-p_h)/\omega$, and
$(1-\gamma)\nu-(1-p_h\gamma)D=-\gamma(1-p_h)$,
\[
\frac{\bar\pi'}{\bar\pi}=\frac{1-p_h}{\omega}+\frac{(1-\gamma)\nu-(1-p_h\gamma)D}{\nu D}
=(1-p_h)\,\frac{\nu D-\gamma\omega}{\nu D\,\omega}.
\]
Expanding gives $\nu D-\gamma\omega=F$. Multiplying by $\bar\pi=\nu q/D$ and using
$1/\omega=q/(p_hM)$ gives \eqref{eq:pibar_prime}. The endpoint values follow from
$\nu(0)=D(0)=1$, $\omega(0)=p_hM+1-p_h$, $\nu(1)=p_h\gamma$, $D(1)=\gamma$ and
$\omega(1)=p_hM$. Moreover $F'(r)=2(1-\gamma)\bigl[(1-p_h\gamma)r-1\bigr]<0$ on
$[0,1]$.

The factor multiplying $F$ in \eqref{eq:pibar_prime} is positive and $F$ is
strictly decreasing, so $\bar\pi$ increases while $F>0$ and decreases once $F<0$.
This gives the three cases. In the interior case, the larger root of $F$ is at
least its vertex $1/(1-p_h\gamma)>1$, so $r_{\bar\pi}$ is the smaller root. Its
radicand, $(1-\gamma)^2-(1-\gamma)(1-p_h\gamma)(1-\gamma+p_h\gamma-p_h\gamma M)$,
equals $(1-\gamma)p_h\gamma\bigl[(1-p_h\gamma)M-\gamma(1-p_h)\bigr]$, because
$(1-p_h\gamma)(1-\gamma+p_h\gamma)=1-\gamma+p_h\gamma^2-p_h^2\gamma^2$. Dividing
the numerator and denominator by $1-\gamma$ gives \eqref{eq:r_pibar}.
\end{proof}

\section{Finite-step bound}\label{app:finite_rho}
\begin{corollary}
\label{cor:finite_rho} For every $\rho\in(0,1)$, every policy $\sigma$ in the sense of
\eqref{eq:policy_def}, and every initial condition $r_0\in[0,1]$,
\begin{equation}
\mathcal J(\sigma,r_0)\;\le\;\max_{r\in[0,1]}\bar\pi(r)+\frac{C_\chi}{2}\,\rho , \quad \text{where} \quad C_\chi:=\frac{1-p_h}{\gamma\,p_h M}+\frac{1-\gamma}{\gamma^2}.
\label{eq:finite_rho_bound}
\end{equation}

\end{corollary}

\begin{proof}
 Let $\chi=-q/D$ and
$K(r)=\int_0^r\chi(v)\,dv$ be as in the proof of Theorem~\ref{thm:fluid}, with
$D$ as in \eqref{eq:u_hat}. We show that
$\sup_{r\in[0,1]}|\chi'(r)|=C_\chi$. Since $D'=-(1-\gamma)$, the quotient rule
gives
\begin{equation}
    \chi'(r)=-\frac{q'(r)D(r)+(1-\gamma)q(r)}{D(r)^{2}},
\qquad\text{so}\qquad
|\chi'(r)|=\frac{q'(r)}{D(r)}+\frac{(1-\gamma)\,q(r)}{D(r)^{2}} .
\end{equation}
From \eqref{eq:probabilistic_purchase}, $q'(r)=ab/\bigl(a+b(1-r)\bigr)^{2}$
with $a:=p_h(1+c_o)$ and $b:=(1-p_h)c_l$. The functions $q$, $q'$ and $1/D$
are nonnegative and nondecreasing on $[0,1]$, so $|\chi'|$ is nondecreasing
and its supremum is attained at $r=1$. Using $q(1)=1$,
$q'(1)=b/a=(1-p_h)/(p_hM)$ and $D(1)=\gamma$ gives $|\chi'(1)|=C_\chi$ of \eqref{eq:finite_rho_bound}. In particular, $K$ is twice continuously
differentiable on $[0,1]$ with $K'=\chi$ and $|K''|\le C_\chi$.

 Fix any policy, with lying rates $u_t$ as in
\eqref{eq:policy_def}. Write
$\Delta_t:=r_{t+1}-r_t=\rho(\lambda_t-r_t)$, so that $|\Delta_t|\le\rho$. By
Taylor's theorem,
$K(r_{t+1})-K(r_t)=\chi(r_t)\Delta_t+\tfrac12K''(\xi_t)\Delta_t^{2}$ for some
$\xi_t$ between $r_t$ and $r_{t+1}$. Let
$R_t:=\tfrac12\E[K''(\xi_t)\,\Delta_t^2\mid\mathcal F_t]$; since $|K''|\le C_\chi$,
$|R_t|\le\tfrac12C_\chi\rho^2$. Taking conditional expectations,
\begin{equation}
    \E\bigl[K(r_{t+1})-K(r_t)\mid\mathcal F_t\bigr]
=\rho\,\chi(r_t)\,b(r_t,u_t)+R_t .
\end{equation}
Combining with \eqref{eq:affine_identity} and
$\bar\pi(r_t)\le\max_r\bar\pi(r)$,
\begin{equation}
    \E[a_t\mid\mathcal F_t]
=\bar\pi(r_t)+\chi(r_t)\,b(r_t,u_t)
\le\max_r\bar\pi(r)
+\frac1\rho\,\E\bigl[K(r_{t+1})-K(r_t)\mid\mathcal F_t\bigr]
+\frac{C_\chi}{2}\rho .
\end{equation}
Taking expectations, averaging over $t<T$ and telescoping,
\begin{equation}
    \frac1T\sum_{t<T}\E[a_t]
\le\max_r\bar\pi(r)+\frac{\E[K(r_T)]-K(r_0)}{\rho T}+\frac{C_\chi}{2}\rho
\le\max_r\bar\pi(r)+\frac{2}{\gamma\rho T}+\frac{C_\chi}{2}\rho ,
\end{equation}
using $|K|\le1/\gamma$ from the proof of Theorem~\ref{thm:fluid}. Letting
$T\to\infty$ gives \eqref{eq:finite_rho_bound}.
\end{proof}

Corollary~\ref{cor:finite_rho} is an upper bound on the unapproximated
chain, so it applies in particular to the $(s,S)$ class, which is not
stationary in $r_t$ alone, and it does not use the regenerative approximation
of Theorem~\ref{thm:sS_collapse_threshold}. It implies that no policy can
exceed $\max_r\bar\pi(r)$ by more than $O(\rho)$, which is consistent with the gaps measured in Section~\ref{sec:num_classes}.

\end{document}